\documentclass[a4paper,USenglish]{lipics}
\def\longVersion{}
  \usepackage{microtype}
  \usepackage{article}
  \usepackage{rev}

  \usepackage[inline]{enumitem}
  \setlist[itemize]{noitemsep,topsep=0pt}

  \usepackage{booktabs}

  \usepackage{dcolumn}
  \newcolumntype{d}[1]{D{.}{.}{#1} }

  \allowdisplaybreaks[4]

  \usepackage[numbers]{natbib}

  \usepackage{hyperref}

\titlerunning{Derived-term Automata for Extended Weighted Rational Expressions}
\title{Derived-term Automata for\newline
  Extended Weighted Rational Expressions}
\author{Akim Demaille {\scriptsize (\SvnDate \SvnRev)}}
\affil{%
 LRDE, EPITA, \texttt{akim@lrde.epita.fr}%
}
\EventShortName{}

\begin{document}
\maketitle

\ifthenelse{\boolean{long}}{\tableofcontents}

\begin{abstract}
  We present an algorithm to build an automaton from a rational expression.
  This approach introduces support for extended weighted expressions.
  Inspired by derived-term based algorithms, its core relies on a different
  construct, \emph{rational expansions}.  We introduce an inductive
  algorithm to compute the expansion of an expression from which the
  automaton follows.  This algorithm is independent of the size of the
  alphabet, and actually even supports infinite alphabets.  It can easily be
  accommodated to generate deterministic (weighted) automata.  These
  constructs are implemented in \vcsn, a free-software platform dedicated to
  weighted automata and rational expressions.
\end{abstract}

\section{Introduction}

Foundational to Automata Theory, the Kleene Theorem (and its weighted
extension, the Kleene--Sch\"utzenberger Theorem) states the equivalence of
\dfn{recognizability} ---accepted by an automaton--- and \dfn{rationality}
---defined by a \dfn{rational}, or \dfn{regular}, expression.  Numerous
constructive proofs (read \emph{algorithms}) have been proposed to go from
rational expressions to automata, and vice versa.  This paper focuses on
building an automaton from an expression.

In 1961 \citet{glushkov.61.rms} provides an algorithm to build a
nondeterministic automaton (without spontaneous transitions) now often
called the standard (or position, or Glushkov) automaton.
Earlier (1960), \citet{mcnaughton.60.itec} proposed the same construct for
\emph{extended} rational expressions (i.e., including intersection and
complement operators), but performed the now usual subset-automaton
construction on-the-fly, thus yielding a deterministic automaton.
A key ingredient of these algorithms is that they build an automaton whose
states represent positions in the rational expression, and computations on
these automata actually represent ``executions'' of the rational expression.

\ifthenelse{\boolean{long}}{Similarly, in}{In} 1964
\citet{brzozowski.64.jacm} shows that \emph{extended} expressions can be
used directly as acceptors: transitions are ``performed'' by computing the
left-quotient of the current expression by the current letter.  With a
proper equivalence relation between expressions (namely ACI: associativity,
commutativity, and idempotence of the addition), Brzozowski shows that there
is a finite number of equivalence classes of such quotients, called
\dfn{derivatives}.  This leads to a very natural construction of a
\emph{deterministic} automaton whose states are these derivatives.  A rather
discreet sentence (last line of p.~484) introduces the concept of
``expansion'', which is not further developed.

In 1996 \citet{antimirov.1996.tcs} introduces a novel idea: do not apply ACI
equivalence globally; rather, when computing the derivative of an expression
which is a sum, split it in a set of ``partial derivatives'' (or ``derived
terms'') --- which amounts to limiting ACI to the sums that are at the root
of the expression.  A key feature of the built automaton is that it is
non-deterministic; as a result the worst-case size of resulting automaton is
linear in the size of the expression, instead of exponential with
Brzozowski's construct.  Antimirov also suggests \emph{not} to rely on
derivation in implementations, but on so called ``linear forms'', which are
closely related to Brzozowski's expansions; derivation is used only to prove
correctness.

In 2005 \citet{lombardy.2005.tcs} generalize the computation of the
derivation and derived-term automaton to support weights.  Since, as is
well-known, not all weighted non deterministic automata can be determinized,
their construct relies on a generalization of Antimirov's derived-term that
generates a \emph{non-deterministic} automaton.  In their formalization,
Antimirov's sets of derived terms naturally turn into \emph{weighted} sets
---each term is associated with a weight--- that they name
\emph{polynomials} (of expressions).  However, linear forms completely
disappear, and the construction of the derived-term automaton relies on
derivatives.  Independently, and with completely different foundations,
\citet{rutten.1999.icalp, rutten.2003.tcs} proposes a similar construction.

In 2011, \citet{caron.2011.lata.2} complete Antimirov's construct to support
extended expressions.  This is at the price of a new definition of
derivatives: sets of sets of expressions, interpreted as disjunctions of
conjunctions of expressions.

\smallskip

The contributions of this paper are threefold.  Firstly, we introduce
``expansions'', which generalize Brzozowski's expansions and Antimirov's
linear forms to support weighted expressions; they bind together the
derivatives, the constant terms and the ``firsts'' of an expression.  They
make the computation of the derived-term automaton independent of the size
of the alphabet, and actually completely eliminate the need for the alphabet
to be finite.  Secondly, we provide support for extended weighted rational
expressions, which generalizes both \citet{lombardy.2005.tcs} and
\citet{caron.2011.lata.2}. And thirdly, we introduce a variation of this
algorithm to build \emph{deterministic} (weighted) automata.

We first settle the notations in \cref{sec:notations}, provide an algorithm
to compute the expansion of an expression in \cref{sec:expr-to-expa}, which
is used in \cref{sec:expaton} to propose an alternative construction of the
derived-term automaton.  In \cref{sec:related} we expose related work and
conclude in \cref{sec:conc}.

\smallskip

Interested readers may experiment with the concepts introduced here using
\vcsn.  \vcsn is a free-software platform dedicated to weighted automata and
rational expressions \citep{demaille.13.ciaa}.  It supports both derivations
and expansions, as exposed in this paper, and the corresponding
constructions of the derived-term automaton\footnote{\label{foot:url}See the
  interactive
  environment, \url{http://vcsn-sandbox.lrde.epita.fr}, or its documentation,\\
  \url{http://vcsn.lrde.epita.fr/dload/2.2/notebooks/expression.derived_term.html}.}.

\section{Notations}
\label{sec:notations}

\newcommand{\OB}[2]{\overbrace{#2}^{\text{#1}}}
\newcommand{\POB}[2]{\OB{\makebox[0pt]{#1}}{\vphantom{\bra{2}}#2}}
\newcommand{\UB}[2]{\underbrace{#2}_{\text{#1}}}
\newcommand{\PUB}[2]{\UB{\makebox[0pt]{#1}}{\vphantom{\bra{2}}#2}}

Our purpose is to define, compute, and use \emph{rational expansions}.  They
intend to be to the differentiation (derivation) of rational expressions
what differential forms are to the differentiation of functions.  Defining
expansions requires several concepts, defined bottom-up in this section.
The following figure should help understanding these different entities, how
they relate to each other, and where we are heading to: given a weighted
rational expression $\Ed_1 = \EdOne$ (weights are written in angle
brackets), compute its expansion:
\vspace{-.8\baselineskip}
\begin{align*}
\UB{Expansion (\cref{sec:expa})}
{
  \PUB{Constant term}
  {
    \vphantom{\PUB{Dummy}{\POB{Weight}{\biggl[bra{2}\biggr]}}}
    \POB{Weight}{\bra{5}}
  }
  \;
  \oplus
  \PUB{Proper part of the expansion}
  {
    \UB{First}{
      \vphantom{\biggl[\biggr]}\POB{Letter}{a}
    }
    \odot
    \biggl[
      \Lmul{2}{\UB{Derived term}{\vphantom{\biggl[\biggr]}\POB{Expression (\cref{sec:expr})}{ce}}}
      \;\oplus\;\;\;\;
      \POB{Monomial}{
        \Lmul{4}{de}
        }
    \biggr]
    \;\;\;\oplus\;\;\;
    b \odot
    \biggl[
      \OB{Polynomial (\cref{sec:poly})}
      {
        \Lmul{6}{\vphantom{\biggl[\biggr]}ce}
        \;\oplus\;
        \Lmul{3}{de}
      }
    \biggr]
  }
}
\end{align*}
It is helpful to think of expansions as a normal form for expressions.

\subsection{Rational Series}

Series are to weighted automata what languages are to Boolean automata.  Not
all languages are rational (denoted by an expression), and similarly, not
all series are rational (denoted by a weighted expression).  We follow
\citet{sakarovitch.09.eat}.

Let $A$ be a (finite) alphabet, and $\bra{\K, +, \cdot, \zeK, \unK}$ a
semiring whose (possibly non commutative) multiplication will be denoted by
implicit concatenation.  A (formal power) \dfn{series} over $A^*$ with
\dfn{weights} (or \dfn{multiplicities}) in $\K$ is any map from $A^*$ to
$\K$.  The weight of a word $m$ in a series $s$ is denoted $s(m)$.  The
\dfn{support} of a series $s$ is the language of words that have a non-zero
weight in $s$.  The \dfn{empty} series, $m \mapsto \zeK$, is denoted $0$;
for any word $u$ (including $\eword$), $u$ denotes the series
$m \mapsto \unK \text{ if $m = u$}, \zeK \text{ otherwise}$.  Equipped with
the pointwise addition ($s + t \coloneqq m \mapsto s(m) + t(m)$) and the
Cauchy product
($s \cdot t \coloneqq m \mapsto \sum_{u, v\in A^* \mid u \cdot v = m} s(u)
\cdot t(v)$)
as multiplication, the set of these series forms a semiring denoted
$\bra{\SRka, +, \cdot, 0, \eword}$.

The \dfn{constant term} of a series $s$, denoted $s_\eword$, is $s(\eword)$,
the weight of the empty word.  A series $s$ is \dfn{proper} if
$s_\eword = \zeK$. The \dfn{proper part} of $s$, denoted $s_p$, is the
proper series which coincides with $s$ on non empty words:
$s = s_\eword + s_p$.

The \dfn{star} of a series is an infinite sum: $s^* \coloneqq \sum_{n\in\N}s^n$.
To ensure semantic soundness, we suppose that $\K$ is a
\dfn{topological semiring}, i.e., it is equipped with a topology, and both
addition and multiplication are continuous.  Besides, it is supposed to be
\dfn{strong}, i.e., the product of two summable families is summable.  This
ensures that $\SRka$, equipped with the product topology derived from the
topology on $\K$, is also a strong topological semiring.

\begin{Proposition}
  \label{prop:dev}
  Let $\K$ be a strong topological semiring. Let $s \in \SRka$,
  $s^*$ is defined iff $s_\eword^*$ is defined and then
  $s^* = s_\eword^* + s_\eword^*s_ps^*$.
\end{Proposition}
\begin{proof}
  By \citep[Prop.~2.6, p.~396]{sakarovitch.09.eat} $s^*$ is defined iff
  $s_\eword^*$ is defined and then
  $s^* = (s_\eword^*s_p)^*s_{\eword}^* = s_\eword^*(s_ps_\eword^*)^*$.  The
  result then follows directly from $s^* = \eword + ss^*$:
  $s^* = s_\eword^*(s_ps_\eword^*)^* = s_\eword^*(\eword +
  (s_ps_\eword^*)(s_ps_\eword^*)^*) = s_\eword^* +
  s_\eword^*s_p(s_\eword^*(s_ps_\eword^*)^*) = s_\eword^* +
  s_\eword^*s_ps^*$.
\end{proof}

Rational languages are closed under intersection.  Series support a natural
generalization of intersection, the Hadamard product, which we name
\dfn{conjunction} and denote $\&$.  The conjunction of series $s$ and $t$ is
defined as $s \AND t \coloneqq m \mapsto s(m) \cdot t(m)$.

Rational languages are also closed under complement, but generalizing this
concept to series is more debatable.  In the sequel, we will rely on the
following definition: ``$s^c$ is the characteristic series of the complement
of the support of $s$.''  More precisely, $s^c(m) \coloneqq s(m)^c$ where
$\forall k \in \K, k^c \coloneqq \unK$ if $k = \zeK$, $\zeK$ otherwise.

\begin{Proposition}
  \label{prop:series}
  For series $s, s', t, t', s_a, t_a \in \SRka$ with $a \in A$, for
  $S, T\subseteq A$, and weights $k, h, s_{\eword}, t_{\eword} \in \K$:
  \begin{gather}
    \label{eq:series:and:distrib}
    (s+s') \AND t = s\AND t + s'\AND t
    \e
    s \AND (t+t') = s\AND t + s\AND t'
    \e
    (ks) \AND (ht) = (kh) (s\AND t)
    \\
    \label{eq:series:and:zip}
    \parenBig{s_{\eword} + \sum_{\mathclap{a\in S}} a\cdot s_a}
    \AND
    \parenBig{t_{\eword} + \sum_{\mathclap{a\in T}} a\cdot t_a}
    = s_{\eword}t_{\eword} + \sum_{\mathclap{a\in S \cap T}} a\cdot(s_a \AND t_a)
    \\
    \label{eq:series:compl}
    \parenBig{s_{\eword} + \sum_{\mathclap{a\in S}} a\cdot s_a}^{c}
    = s_{\eword}^{c} + \sum_{\mathclap{a\in S}} a\cdot s_a^{c} + \sum_{\mathclap{a\in A\setminus S}} a\cdot 0^{c}
  \end{gather}
\end{Proposition}

\subsection{Extended Weighted Rational Expressions}
\label{sec:expr}
\begin{Definition}[Extended Weighted Rational Expression]
  A \dfn{rational (or regular) expression} $\Ed$ is a term built from the
  following grammar, where $a \in A$ is a letter, and $k \in \K$ a weight:
  \begin{math}
    \Ed \Coloneqq \zed
          \mid \und
          \mid a
          \mid \Ed + \Ed
          \mid \lmul{k}{\Ed}
          \mid \rmul{\Ed}{k}
          \mid \Ed \cdot \Ed
          \mid \Ed^*
          \mid \Ed \AND \Ed
          \mid \Ed^c
  \end{math}.
\end{Definition}

Since the product of $\K$ does not need to be commutative there are two
exterior products: $\lmul{k}{\Ed}$ and $\rmul{\Ed}{k}$.
The \dfn{size} (aka \dfn{length}) of an expression $\Ed$, $\length{\Ed}$, is
its number of symbols, excluding parentheses; its \dfn{width} (aka
\dfn{literal length}), $\width{\Ed}$, is the number of occurrences of
letters.

Rational expressions are syntactic objects; they provide a finite notations
for (some) series, which are semantic objects.
\begin{Definition}[Series Denoted by an Expression]
  Let $\Ed$ be an expression.  The series denoted by $\Ed$, noted
  $\sem{\Ed}$, is defined by induction on $\Ed$:
  \begin{gather*}
    \sem{\zed} \coloneqq 0 \qquad
    \sem{\und} \coloneqq \eword    \qquad
    \sem{a}    \coloneqq a \qquad
    \sem{\Ed+\Fd}       \coloneqq \sem{\Ed} + \sem{\Fd} \qquad
    \sem{\lmul{k}{\Ed}} \coloneqq {k}{\sem{\Ed}}   \\
    \sem{\rmul{\Ed}{k}} \coloneqq {\sem{\Ed}}{k}    \quad
    \sem{\Ed \cdot \Fd} \coloneqq \sem{\Ed} \cdot \sem{\Fd} \quad
    \sem{\Ed^*} \coloneqq \sem{\Ed}^* \quad
    \sem{\Ed \AND \Fd} \coloneqq \sem{\Ed} \AND \sem{\Fd}  \quad
    \sem{\Ed^c} \coloneqq \sem{\Ed}^c
  \end{gather*}
\end{Definition}
An expression is \dfn{valid} if it denotes a series.  More specifically,
this requires that $\sem{\Fd}^*$ is well defined for each subexpression of
the form $\Fd^*$, i.e., that the constant term of $\sem{\Fd}$ is
\emph{starrable} in $\K$ (\cref{prop:dev}).  This definition, which involves
series (semantics) to define a property of expressions (syntax), will be
made effective (syntactic) with the appropriate definition of the constant
term $c(\Ed)$ \emph{of an expression $\Ed$} (\cref{def:ctder}).

\begin{longenv}
\begin{Example}[{\citep[Example~1]{lombardy.2005.tcs}}]
  \label{ex:e2}
  Expressions $\Fd_2 \coloneqq \FdTwo, \Ed_2 = \Fd_2^*$ have weights in
  $\Q$.  $\Fd_2$ is valid: its stars are on expressions that denote proper
  series.  $\Ed_2$ is valid, as the constant term of $\sem{\Fd_2}$ is
  $\tfrac{1}{6} + \tfrac{1}{3} = \tfrac{1}{2}$, whose star is defined:
  2. $\length{\Ed_2} = 8, \width{\Ed_2} = 2$.
\end{Example}
\end{longenv}

Two expressions $\Ed$ and $\Fd$ are \dfn{equivalent} iff
$\sem{\Ed} = \sem{\Fd}$.  Some expressions are ``trivially equivalent''; any
candidate expression will be rewritten via the following \dfn{trivial
  identities}.  Any subexpression of a form listed to the left of a
`$\Rightarrow$' is rewritten as indicated on the right.
\begin{gather*}
  \Ed+\zed  \Rightarrow \Ed
  \ee
  \zed+\Ed  \Rightarrow \Ed
  \\
  \begin{aligned}[t]
    \lmul{\zeK}{\Ed} & \Rightarrow \zed &
    \lmul{\unK}{\Ed} & \Rightarrow \Ed  &
    \lmul{k}{\zed}   & \Rightarrow \zed &
    \lmul{k}{\lmul{h}{\Ed}} &\Rightarrow \lmul{kh}{\Ed}
    \\
    \rmul{\Ed}{\zeK} & \Rightarrow \zed &
    \rmul{\Ed}{\unK} & \Rightarrow  \Ed &
    \rmul{\zed}{k}   & \Rightarrow \zed &
    \rmul{\rmul{\Ed}{k}}{h}  &\Rightarrow \rmul{\Ed}{kh}
  \end{aligned}\\
  \rmul{(\lmul{k}{\Ed})}{h} \Rightarrow \lmul{k}{(\rmul{\Ed}{h})} \ee
  \rmul{\ell}{k} \Rightarrow \lmul{k}{\ell}
  \\ 
  \Ed \cdot \zed  \Rightarrow \zed \ee
  \zed \cdot \Ed  \Rightarrow \zed
  \\
  (\lmulq{k}{\und}) \cdot \Ed   \Rightarrow  \lmul{k}{\Ed}
  \ee
  \Ed \cdot (\lmulq{k}{\und})   \Rightarrow  \rmul{\Ed}{k}
  \\ 
  \zed^\star \Rightarrow \und
  \\
  \Ed \AND \zed  \Rightarrow \zed  \ee \zed \AND \Ed  \Rightarrow \zed
  \ee
  \Ed \AND \zed^c  \Rightarrow \Ed \ee \zed^c \AND \Ed  \Rightarrow \Ed
  \\
  \lmulq{k}{\ell} \AND \lmulq{h}{\ell}  \Rightarrow  \lmul{kh}{\ell}
  \ee
  \lmulq{k}{\ell} \AND \lmulq{h}{\ell'}  \Rightarrow  \zed
  \\
  (\lmul{k}{\Ed})^c \Rightarrow \Ed^c \ee
  (\rmul{\Ed}{k})^c \Rightarrow \Ed^c
\end{gather*}
where $\Ed$ stands for a rational expression, $a \in A$~is a letter,
$\ell, \ell' \in A \cup \{\und\}$ denote two different labels, $k, h\in \K$
are weights, and $\lmulq{k}{\ell}$ denotes either $\lmul{k}{\ell}$, or
$\ell$ in which case $k = \unK$ in the right-hand side of $\Rightarrow$.
The choice of these identities is beyond the scope of this paper (see
\cite{sakarovitch.09.eat}), however note that, with the exception of the
last line, they are limited to trivial properties; in particular
\dfn{linearity} (``weighted ACI'': associativity, commutativity, and
$\lmul{k}{\Ed} + \lmul{h}{\Ed} \Rightarrow \lmul{k+h}{\Ed}$) is not
enforced.  In practice, additional identities help reducing the number of
derived terms \citep{owens.2009.jfp}, hence the final automaton size.  The
last two rules, about complement, will be discussed in \cref{sec:compl};
they are disabled when $\K$ has zero divisors.

\begin{Example}
  \label{ex:ab}
  Conjunction and complement can be combined to define new operators which
  are convenient syntactic sugar.  For instance,
  $\Ed \lplus \Fd \coloneqq \Ed + (\Ed^c \AND \Fd)$ allows to define a
  left-biased $+$ operator: $\sem{\Ed \lplus \Fd}(u) = \sem{\Ed}(u)$ if
  $\sem{\Ed}(u) \neq \zeK$, $\sem{\Fd}(u)$ otherwise.  The following example
  mocks Lex-like scanners: identifiers are non-empty sequences of letters of
  $\{a, b\}$ that are not reserved keywords.  The expression
  $\Ed_3 \coloneqq \lmul{2}{ab} \lplus \lmul{3}{(a+b)^+}$, with weights in
  $\mathbb{Z}$, maps the ``keyword'' $ab$ to 2, and ``identifiers'' to 3.
  Once desugared and simplified by the trivial identities, we have
  $\Ed_3 = \lmul{2}{ab} + ((ab)^{c} \AND \lmul{3}{((a+b)(a+b)^*))}$.
\end{Example}

\subsection{Rational Polynomials}
\label{sec:poly}

At the core of the idea of ``partial derivatives'' introduced by
\citet{antimirov.1996.tcs}, is that of \emph{sets} of rational expressions,
later generalized in \emph{weighted sets} by \citet{lombardy.2005.tcs},
i.e., functions (partial, with finite domain) from the set of rational
expressions into $\K \setminus \{\zeK\}$.  It proves useful to view such
structures as ``polynomials of rational expressions''.  In essence, they
capture the linearity of addition.

\begin{Definition}[Rational Polynomial]
  A \dfn{polynomial} (of rational expressions) is a finite (left) linear
  combination of rational expressions.  Syntactically it is represented by a
  term built from the grammar
  \begin{math}
    \Pd \Coloneqq 0
    \mid \Lmul{k_1}{\Ed_1} \oplus \cdots \oplus \Lmul{k_n}{\Ed_n}
  \end{math}
  where $k_i\in \K\setminus \{\zeK\}$ denote \emph{non-null} weights, and
  $\Ed_i$ denote \emph{non-null} expressions.  Expressions may not appear
  more than once in a polynomial.  A \dfn{monomial} is a pair
  $\bra{k_i} \odot \Ed_i$.
\end{Definition}

We use specific symbols ($\odot$ and $\oplus$) to clearly separate the outer
polynomial layer from the inner expression layer.
A polynomial $\Pd$ of rational expressions can be ``projected'' as a
rational expression $\expr{\Pd}$ by mapping its sum and left-multiplication
by a weight onto the corresponding operators on rational expressions.  This
operation is performed on a canonical form of the polynomial (expressions
are sorted in a well defined order).  Polynomials denote series:
$\sem{\Pd} \coloneqq \sem{\expr{\Pd}}$.

\begin{Example}\label{ex:e1}
  Let $\Ed_1 \coloneqq \EdOne$.  Polynomial
  `$\Pd_{1a} \coloneqq \Lmul{2}{ce} \oplus \Lmul{4}{de}$' has two monomials:
  `$\Lmul{2}{ce}$' and `$\Lmul{4}{de}$'.  It denotes the (left) quotient of
  $\sem{\Ed_1}$ by $a$, and
  `$\Pd_{1b} \coloneqq \Lmul{6}{ce} \oplus \Lmul{3}{de}$' the quotient by
  $b$.
\end{Example}

Let $\Pd = \bra{k_1} \odot \Ed_1 \oplus \cdots \oplus \bra{k_n} \odot \Ed_n$
be a polynomial, $k$ a weight (possibly null) and $\Fd$ an expression
(possibly null), we introduce the following operations:
\begin{gather}
  \notag
  \Pd\cdot\Fd
  \coloneqq
    \bra{k_1} \odot (\Ed_1\cdot\Fd) \oplus \cdots \oplus \bra{k_n} \odot (\Ed_n\cdot\Fd)
  \\
  \notag
  \lmul{k}{\Pd}
   \coloneqq
    \bra{kk_1} \odot \Ed_1 \oplus \cdots \oplus \bra{kk_n} \odot \Ed_n
  \ee
  \rmul{\Pd}{k}
  \coloneqq
  \bra{k_1} \odot (\rmul{\Ed_1}{k}) \oplus \cdots \oplus \bra{k_n} \odot (\rmul{\Ed_n}{k})
  \\\label{eq:poly:ops:andcompl}
  \Pd_1 \AND \Pd_2
  \coloneqq
  \bigoplus_{{\substack{\Lmul{k_1}{\Ed_1} \in \Pd_1\\\Lmul{k_2}{\Ed_2} \in \Pd_2}}} \Lmul{k_1k_2}{(\Ed_1 \AND \Ed_2)}
  \ee
  \Pd^{c}
  \coloneqq
  \Lmul{\unK}{\expr{\Pd}^{c}}
\end{gather}
Trivial identities might simplify the result, e.g.,
$(\Lmul{\unK}{a}) \AND (\Lmul{\unK}{b}) = \Lmul{\unK}{(a \AND b)} = 0$.

Note the asymmetry between left and right exterior products.  The addition
of polynomials is commutative, multiplication by zero (be it an expression
or a weight) evaluates to the null polynomial, and the left-multiplication
by a weight is distributive.


\begin{Lemma}
  \label{lem:poly:ops}%
  \begin{math}
    \sem{\Pd\cdot\Fd} = \sem{\Pd} \cdot \sem{\Fd}
    \ee
    \sem{\lmul{k}{\Pd}} = \lmul{k}{\sem{\Pd}}
    \ee
    \sem{\rmul{\Pd}{k}} = \rmul{\sem{\Pd}}{k}
  \end{math}
  \\
  \begin{math}
    \sem{\Pd_1 \AND \Pd_2} = \sem{\Pd_1} \AND \sem{\Pd_2}
    \ee
    \sem{\Pd^{c}} = \sem{\Pd}^{c}
  \end{math}.
\end{Lemma}
\begin{proof}
  The first three are trivial.  The case of $\AND$ follows from
  \cref{eq:series:and:distrib}. Complement follows from its definition:
  $\sem{\Pd^c} \coloneqq \sem{\expr{\Pd^c}} =
  \sem{\Lmul{\unK}{\expr{\Pd}^{c}}} = \sem{\expr{\Pd}^{c}} =
  \sem{\expr{\Pd}}^{c} = \sem{\Pd}^c$.
\end{proof}

\subsection{Rational Expansions}
\label{sec:expa}


\begin{Definition}[Rational Expansion]
  A \dfn{rational expansion} $\Xd$ is a term built from the grammar
  \begin{math}
    \Xd \Coloneqq \bra{k} \oplus a_1 \odot[\Pd_1] \oplus \cdots \oplus a_n \odot[\Pd_n]
  \end{math}
  where $k \in \K$ is a weight (possibly null), $a_i \in A$ letters
  (occurring at most once), and $\Pd_i$ non-null polynomials.  We name $k$
  the \dfn{constant term},
  $a_1 \odot[\Pd_1] \oplus \cdots \oplus a_n \odot[\Pd_n]$ the \dfn{proper
    part}, and $\{a_1, \ldots, a_n\}$ (possibly empty) the \dfn{firsts}.
\end{Definition}
To ease reading, polynomials are written in square brackets.  Contrary to
expressions and polynomials, there is no specific term for the empty
expansion: it is represented by $\bra{\zeK}$, the null weight.  Except for
this case, null constant terms are left implicit.  Besides their support for
weights, expansions differ from Antimirov's linear forms in that they
integrate the constant term, which gives them a flavor of series.  Given an
expansion $\Xd$, we denote by $\Xd_\eword$ (or $\Xd(\eword)$) its constant
term, by $f(\Xd)$ its firsts, by $\Xd_p$ its proper part, and by $\Xd_a$ (or
$\Xd(a)$) the polynomial corresponding to $a$ in $\Xd$.  Expansions will
thus be written:
\begin{math}
  \Xd = \bra{\Xd_{\eword}} \oplus \bigoplus_{a \in f(\Xd)} a \odot[\Xd_a]
\end{math}.


An expansion whose polynomials are monomials is said to be
\dfn{deterministic}.  An expansion $\Xd$ can be ``projected'' as a rational
expression $\expr{\Xd}$ by mapping weights, letters and polynomials to their
corresponding rational expressions, and $\oplus$/$\odot$ to the
sum/concatenation of rational expressions.  Again, this is performed on a
canonical form of the expansion: letters and polynomials are sorted.
Expansions also denote series: $\sem{\Xd} \coloneqq \sem{\expr{\Xd}}$.  An
expansion $\Xd$ is said to be \dfn{equivalent} to an expression $\Ed$ iff
$\sem{\Xd} = \sem{\Ed}$.

\begin{Example}[\cref{ex:e1} continued]
  \label{ex:e1:xpn}
  Expansion
  $\Xd_1 \coloneqq \bra{5} \oplus a \odot [\Pd_{1a}] \oplus b \odot
  [\Pd_{1b}]$
  has $\Xd_1(\eword) = \bra{5}$ as constant term, and maps the letter $a$
  (resp.\ $b$) to the polynomial $\Xd_1(a) = \Pd_{1a}$ (resp.\
  $\Xd_1(b) = \Pd_{1b}$).  $\Xd_1$ can be proved to be equivalent to
  $\Ed_1$.
\end{Example}

Let $\Xd, \Yd$ be expansions, $k$ a weight, and $\Ed$ an expression (all
possibly null):
\begin{gather}
  \label{eq:epn:plus:epn}
  \Xd \oplus \Yd
  \coloneqq
  \bra{\Xd_\eword+ \Yd_\eword}\oplus\bigoplus_{\mathclap{a \in f(\Xd) \cup f(\Yd)}} a \odot [\Xd_a \oplus \Yd_a]
  \\
  \lmul{k}{\Xd} \coloneqq \bra{k\Xd_\eword} \oplus \bigoplus_{\mathclap{a \in f(\Xd)}} a \odot [\lmul{k}{\Xd_a}]
  \ee
  \rmul{\Xd}{k} \coloneqq \bra{\Xd_\eword k} \oplus \bigoplus_{\mathclap{a \in f(\Xd)}} a \odot [\rmul{\Xd_a}{k}]
  \\
  \label{eq:epn:mul:epn}
  \Xd\cdot\Ed \coloneqq \bigoplus_{a \in f(\Xd)} a \odot [\Xd_a \cdot \Ed] \ee \text{with $\Xd$ proper: $\Xd_\eword = \zeK$}
  \\
  \label{eq:epn:and:epn}
  \Xd \AND \Yd \coloneqq \bra{\Xd_\eword\Yd_\eword}
    \oplus
    \bigoplus_{\mathclap{a \in f(\Xd) \cap f(\Yd)}} a \odot [\Xd_a \AND \Yd_a]
  \\
  \label{eq:epn:compl:epn}
  \Xd^{c} \coloneqq \bra{\Xd_\eword^{c}}
  \oplus \bigoplus_{\mathclap{a \in f(X)}} a \odot [\Xd_a^{c}] \oplus \bigoplus_{\mathclap{a \in A\setminus f(X)}} a \odot [\zed^{c}]
\end{gather}
Since by definition expansions never map to null polynomials, some firsts
might be smaller that suggested by these equations.  For instance in
$\mathbb{Z}$ the sum of $\bra{1} \oplus a \odot [\Lmul{1}{b}]$ and
$\bra{1} \oplus a \odot [\Lmul{-1}{b}]$ is $\bra{2}$, and
$\paren{a \odot[\Lmul{1}{b}]} \AND \paren{a \odot[\Lmul{1}{c}]}$ is
$\bra{0}$ since $b\AND c \Rightarrow \zed$.  Note that $\Xd^c$ is a
deterministic expansion.

The following lemma is simple to establish: lift semantic equivalences,
such as those of \cref{prop:series}, to syntax, using \cref{lem:poly:ops}.
\begin{Lemma}
  \label{lem:xpn:semantics}%
  \begin{math}
    \sem{\Xd \oplus \Yd} = \sem{\Xd} + \sem{\Yd} \ee
    \sem{\lmul{k}{\Xd}} = \lmul{k}{\sem{\Xd}} \ee
    \sem{\rmul{\Xd}{k}} = \rmul{\sem{\Xd}}{k}
  \end{math}
  \\
  \begin{math}
    \sem{\Xd\cdot\Ed} = \sem{\Xd}\cdot\sem{\Ed}
    \ee
    \sem{\Xd \AND \Yd} = \sem{\Xd} \AND \sem{\Yd}
    \ee
    \sem{\Xd^{c}} = \sem{\Xd}^{c}
  \end{math}.
\end{Lemma}

\subsection{Weighted Automata}

\begin{Definition}[Automaton]
  A \dfn{weighted automaton} $\Ac$ is a tuple $\bra{A, \K, Q, E, I, T}$
  where:
  \begin{itemize}
  \item $A$ (the set of labels) is an alphabet (usually finite),
    \ifthenelse{\boolean{long}}{\item}{and} $\K$ (the set of weights) is a
    semiring,
  \item $Q$ is a set of states,\ifthen{\boolean{long}}{\item} $I$ and $T$
    are the \dfn{initial} and \dfn{final} functions from $Q$ into $\K$,
  \item $E$ is a (partial) function from $Q \times A \times Q$ into
    $\K \setminus \{\zeK\}$;

    its domain represents the \dfn{transitions}:
    $(\mathit{source}, \mathit{label}, \mathit{destination})$.
  \end{itemize}
\end{Definition}
An automaton is \dfn{locally finite} if each state has a finite number of
outgoing transitions ($\forall s\in Q, \{s\} \times A \times Q \cap E$ is
finite).  A \dfn{finite automaton} has a finite number of states.
A \dfn{path} $p$ in an automaton is a sequence of transitions
$(q_0, a_0, q_1)(q_1, a_1, q_2)\cdots(q_n, a_n, q_{n+1})$ where the source
of each is the destination of the previous one; its \dfn{label} is the word
$a_0a_1\cdots a_n$, its \dfn{weight} is
$I(q_0) \otimes E(q_0, a_0, q_1) \otimes \cdots \otimes E(q_n, a_n, q_{n+1})
\otimes T(q_{n+1})$.
The \dfn{evaluation} of word $u$ by a locally finite automaton $\Ac$,
$\Ac(u)$, is the (finite) sum of the weights of all the paths labeled by
$u$, or $\zeK$ if there are no such path.  The \dfn{behavior} of such an
automaton $\Ac$ is the series $\sem{\Ac} \coloneqq u \mapsto \Ac(u)$.  A
state $q$ is \dfn{initial} if $I(q) \neq \zeK$.  A state $q$ is
\dfn{accessible} 
if there is a path from an initial state to $q$.  The \dfn{accessible} part
of an automaton $\Ac$ is the subautomaton whose states are the accessible
states of $\Ac$.  The size of a finite automaton, $\length{\Ac}$, is its
number of states.


We are interested, given an expression $\Ed$, by an algorithm to compute an
automaton $\Ac_\Ed$ such that $\sem{\Ac_\Ed} = \sem{\Ed}$
(\cref{sec:expaton}).  To this end, we first introduce a simple recursive
procedure to compute \emph{the} expansion of an expression.

\section{Computing Expansions of Expressions}
\label{sec:expr-to-expa}

\subsection{Expansion of a Rational Expression}
\label{sec:expa-of-expr}
\begin{Definition}[Expansion of a Rational Expression]
  \label{def:expa-of-expr}
  The \dfn{expansion of a rational expression} $\Ed$, written $d(\Ed)$, is
  the expansion defined inductively as follows:
  \begin{gather}
    \label{eq:epn:cst}
    d(\zed) \coloneqq \bra{\zeK} \ee
    d(\und) \coloneqq \bra{\unK} \ee
    d(a)    \coloneqq a \odot [\Lmul{\unK}{\und}]
    \\
    \label{eq:epn:add}
    d(\Ed+\Fd) \coloneqq d(\Ed) \oplus d(\Fd) \ee
    d(\lmul{k}{\Ed}) \coloneqq \lmul{k}{d(\Ed)} \ee
    d(\rmul{\Ed}{k}) \coloneqq \rmul{d(\Ed)}{k}
    \\
    \label{eq:epn:mul}
    d(\Ed \cdot \Fd)  \coloneqq \ep(\Ed)\cdot \Fd \oplus \lmul{\ec(\Ed)}{d(\Fd)}
    \\
    \label{eq:epn:star}
    d(\Ed^*) \coloneqq \bra{\ec(\Ed)^*} \oplus \lmul{\ec(\Ed)^*}{\ep(\Ed) \cdot \Ed^*}
    \\
    \label{eq:epn:and}
    d(\Ed \AND \Fd) \coloneqq d(\Ed) \AND d(\Fd)
    \\
    \label{eq:epn:compl}
    d(\Ed^{c}) \coloneqq d(\Ed)^{c}
  \end{gather}
  where $\ec(\Ed) \coloneqq d(\Ed)_\eword, \ep(\Ed) \coloneqq d(\Ed)_p$ are
  the constant term/proper part of $d(\Ed)$.
\end{Definition}

The right-hand sides are indeed expansions.  The computation trivially
terminates: induction is performed on strictly smaller subexpressions.
These formulas are enough to compute the expansion of an expression; there
is no secondary process for the firsts --- indeed
$d(a) \coloneqq a \odot [\Lmul{\unK}{\und}]$ suffices and every other case
simply propagates or assembles the firsts --- or the constant terms.  Of
course, in an implementation, a single recursive call to $d(\Ed)$ is
performed for \cref{eq:epn:mul,eq:epn:star}, from which $\ec(\Ed)$ and
$\ep(\Ed)$ are obtained.  So for instance \cref{eq:epn:star} should rather
be written:
\begin{math}
  d(\Ed^*)
  \coloneqq
  \mathtt{let}\;
  \Xd = d(\Ed)\;\mathtt{in}\;
  \bra{\Xd_\eword^*}
  \oplus \lmul{\Xd_\eword^*}{\Xd_p \cdot \Ed^*}
\end{math}.  Besides, existing expressions should be referenced to, not
duplicated: in the previous piece of code, $\Ed^*$ is not built again, the
input argument is reused.

\begin{Proposition}
  The expansion of a rational expression is equivalent to the expression.
\end{Proposition}
\begin{proof}
  We prove that $\sem{d(\Ed)} = \sem{\Ed}$ by induction on the expression.
  The equivalence is straightforward for
  \cref{eq:epn:cst,eq:epn:add}.
  The case of multiplication, \cref{eq:epn:mul}, follows from:
  \begin{align*}
    \sem{d(\Ed \cdot \Fd)}
    &= \sem{\ep(\Ed) \cdot \Fd \oplus \bra{\ec(\Ed)} \cdot d(\Fd)}
    = \sem{\ep(\Ed)}\cdot\sem{\Fd} + \bra{\ec(\Ed)} \cdot \sem{d(\Fd)} \\
    &= \sem{\ep(\Ed)}\cdot\sem{\Fd} + \bra{\ec(\Ed)} \cdot \sem{\Fd}
    = \paren{\sem{\bra*{\ec(\Ed)}}+\sem{\ep(\Ed)}} \cdot\sem{\Fd} \\
    &= \sem{\bra{\ec(\Ed)} + \ep(\Ed)} \cdot\sem{\Fd}
    = \sem{d(\Ed)} \cdot \sem{\Fd}
    = \sem{\Ed} \cdot \sem{\Fd}
    = \sem{\Ed\cdot \Fd}
  \end{align*}
  It might seem more natural to exchange the two terms (i.e.,
  $\bra{\ec(\Ed)} \cdot d(\Fd) \oplus \ep(\Ed)\cdot \Fd$), but an
  implementation first computes $d(\Ed)$ and then computes $d(\Fd)$
  \emph{only if} $\ec(\Ed) \ne \zeK$.
  The case of Kleene star, \cref{eq:epn:star}, follows from \cref{prop:dev}.
  The case of conjunction is straightforward:
\ifthen{\boolean{long}}{
  \begin{align*}
    \sem{d(\Ed \AND \Fd)}
    & = \sem{d(\Ed) \AND d(\Fd)}
    & \text{by definition, \cref{eq:epn:and}}
    \\
    & = \sem{d(\Ed)} \AND \sem{d(\Fd)}
    & \text{by \cref{lem:xpn:semantics}}
    \\
    & = \sem{\Ed} \AND \sem{\Fd}
    & \text{by induction hypothesis}
    \\
    & = \sem{\Ed \AND \Fd}
    & \text{by \cref{lem:xpn:semantics}}
    &  \hfill\qedhere
  \end{align*}
}{
  \begin{align*}
    \sem{d(\Ed \AND \Fd)}
    = \sem{d(\Ed) \AND d(\Fd)}
    = \sem{d(\Ed)} \AND \sem{d(\Fd)}
    = \sem{\Ed} \AND \sem{\Fd}
    = \sem{\Ed \AND \Fd}
    &\hfill\qedhere
  \end{align*}
}
\end{proof}

\subsection{Connection with Derivatives}

We reproduce here the definition of constant terms and derivatives from
Lombardy et al \citep[p.~148 and Def.~2]{lombardy.2005.tcs}, with our
notations and added support for extended expressions.

\begin{Definition}[Constant Term and Derivative]
  \label{def:ctder}
  \begin{align}
    \label{eq:der:cst}
    c(\zed) &\coloneqq \bra{\zeK}, \e
    c(\und) \coloneqq \bra{\unK},
    &
    \da{\zed} &\coloneqq \zed, \e \da{\und} \coloneqq \zed,
    \\
    \label{eq:der:label}
    c(a) &\coloneqq \bra{\zeK}, \forall a \in A, &
    \da{b} &\coloneqq
      \und \text{ if $b = a$, }
      \zed \text{ otherwise,}
    \displaybreak[0]
    \\
    \label{eq:der:add}
    c(\Ed+\Fd) &\coloneqq c(\Ed) + c(\Fd), &
    \da{(\Ed + \Fd)} &\coloneqq \da{\Ed} \oplus \da{\Fd},
    \displaybreak[0]
    \\
    \label{eq:der:lmul}
    c(\lmul{k}{\Ed}) &\coloneqq \lmul{k}{c(\Ed)}, &
    \da{(\lmul{k}{\Ed})} &\coloneqq \lmul{k}{\left(\da{\Ed}\right)},
    \displaybreak[0]
    \\
    \label{eq:der:rmul}
    c(\rmul{\Ed}{k}) &\coloneqq \rmul{c(\Ed)}{k}, &
    \da{(\rmul{\Ed}{k})} &\coloneqq \rmul{\left(\da{\Ed}\right)}{k},
    \displaybreak[0]
    \\
    \label{eq:der:mul}
    c(\Ed \cdot \Fd) &\coloneqq c(\Ed) \cdot c(\Fd), &
    \da{(\Ed \cdot \Fd)} &\coloneqq \left(\da{\Ed}\right)\cdot \Fd \oplus \lmul{c(\Ed)}{\da{\Fd}},
    \\
    \label{eq:der:star}
    c(\Ed^*) &\coloneqq c(\Ed)^*,&
    \da{\Ed^*} &\coloneqq \lmul{c(\Ed)^*}{\left(\da{\Ed}\right)\cdot \Ed^*}
    \\
    \label{eq:der:and}
    c(\Ed \AND \Fd)     & \coloneqq c(\Ed) \cdot c(\Fd), &
    \da{(\Ed \AND \Fd)} & \coloneqq \da{\Ed} \AND \da{\Fd},
    \\
    \label{eq:der:compl}
    c(\Ed^c) &\coloneqq c(\Ed)^c,&
    \da{\Ed^c} &\coloneqq \left(\da{\Ed}\right)^{c}
  \end{align}
  where \cref{eq:der:star} applies iff $c(\Ed)^*$ is defined in $\K$.
\end{Definition}

The reader is invited to compare \cref{def:expa-of-expr} and
\cref{def:ctder}, which does not even include the computation of the firsts.

\begin{Proposition}
  \label{prop:expa:der}
  For any rational expression $\Ed$, $d(\Ed)(\eword) = c(\Ed)$, and
  $d(\Ed)(a) = \da{\Ed}$.
\end{Proposition}

\begin{proof}
  A straightforward induction on $\Ed$.  The cases of constants and letters
  are immediate consequences of \cref{eq:der:cst,eq:der:label} on the one
  hand, and \cref{eq:epn:cst} on the other hand.
  \Cref{eq:epn:add,eq:der:add} both express straightforward ``linearity''.
  Multiplication (concatenation) is again barely a change of notation
  between \cref{eq:epn:mul} and \cref{eq:der:mul}, and likewise for the
  Kleene star (\cref{eq:epn:star,eq:der:star}).  Conjunction,
  \cref{eq:der:and}, follows from \cref{eq:epn:and,eq:epn:and:epn}, and
  complement, \cref{eq:der:compl}, from \cref{eq:epn:compl} and
  \cref{eq:epn:compl:epn}.
\end{proof}

\Cref{prop:expa:der} states that expansions, like Antimirov's linear forms,
offer a different means to compute the expression derivatives.  However
expansions seem to better capture the essence of the process, where the
computations of constant terms are tightly coupled with that of the
derivations.  The formulas are more concise.  Expansions are also ``more
complete'' than derivations, viz., the expansion of an expression can be
seen as a normal-form of this expression: $\Ed \equiv \expr{d(\Ed)}$ and
$d(\Ed)=d(\expr{d(\Ed)})$.  Expansions are more efficient to perform
effective calculations, such building an automaton (\cref{sec:perfs}), while
derivatives are used to prove the correctness (\cref{thm:expaton}).

\section{Expansion-Based Derived-Term Automaton}
\label{sec:expaton}
\vspace{-1ex}
\begin{Definition}[Derived-Term Automaton]
  \label{def:expaton}
  The \dfn{derived-term automaton} of an expression $\Ed$ is the accessible
  part of the automaton ${\Ac_\Ed} \coloneqq \bra{A, \K, Q, E, I, T}$
  defined as follows:
  \begin{itemize}
  \item $Q$ is the set of rational expressions on alphabet $A$ with weights
    in $\K$,
  \item
    $E(\Fd, a, \Fd') = k \text{ iff } a \in f(d(\Fd)) \;\mathrm{and}\;
    \lmul{k}{\Fd'} \in d(\Fd)(a)$,
  \item $I = \Ed \mapsto \unK$, $T(\Fd) = k$ iff $\bra{k} = d(\Fd)(\eword)$.
  \end{itemize}
\end{Definition}
The resulting automaton is locally finite, and not necessarily
deterministic\ifthen{\boolean{long}}{: given a state $\Fd$ and
  $a\in f(d(\Fd))$ one of its firsts, the ``destinations'' are all the
  expressions of $d(\Fd)(a)$}.

\begin{Example}[\cref{ex:e1,ex:e1:xpn} continued]
  \label{ex:e1:end}
\ifthenelse{\boolean{long}}{
  Given $d(\Ed_1)$, $\mathcal{A}_{\Ed_1}$ follows.
  \renewcommand{\UB}[2]{#2}
  \abovedisplayskip=\abovedisplayshortskip
  \begin{align*}
    d(\Ed_1)
    & = \Xd_1 = \bra{2}
      \oplus a \odot \UB{$\Pd_{1a}$}{\left[\Lmul{2}{ce} \oplus \Lmul{4}{de}\right]}
      \oplus b \odot \UB{$\Pd_{1b}$}{\left[\Lmul{6}{ce} \oplus \Lmul{3}{de}\right]}
  \end{align*}}
{We have $d(\Ed_1) = \Xd_1$. $\mathcal{A}_{\Ed_1}$ is:}

  \centerline{\includegraphics[scale=.8]{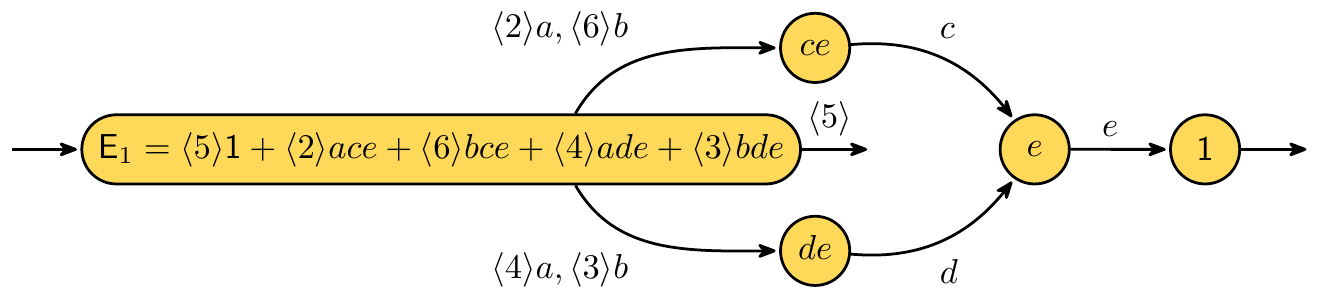}}
\end{Example}

It is straightforward to extract an algorithm from \cref{def:expaton}, using
a work-list of states whose outgoing transitions to compute.  This approach
admits a natural lazy implementation: the whole automaton is not computed at
once, but rather, states and transitions are computed on-the-fly, on demand,
for instance when evaluating a word.

\begin{theorem}
  \label{thm:expaton}%
  Any (valid) expression $\Ed$ and its expansion-based derived-term
  automaton $\Ac_\Ed$ denote the same series, i.e.,
  $\sem{\Ac_\Ed} = \sem{\Ed}$.
\end{theorem}

The smallness of the derived-term automaton for basic operators
($\length{\Ac_\Ed} \le \width{\Ed} + 1$
\citep[Theorem~2]{lombardy.2005.tcs}) no longer applies with extended
operators.  Let $m$ and $n$ be coprime integers,
$\Ed \coloneqq (a^m)^*\AND (a^n)^*$ has width $\width{\Ed} = m+n$; it is
easy to see that $\length{\Ac_\Ed} = mn$.  It is also a classical result
that the minimal (trim) automaton to recognize the language of
$\Fd_n \coloneqq (a+b)^*a(a+b)^n$ has $2^{n+1}$ states; so
$\width{\Fd_n^{c}} = 2n+3$, but $\length{\Ac_{\Fd_n^{c}}} = 2^{n+1}+1$ (the
additional state is the sink state needed to get a \emph{complete}
deterministic automaton before complement).  Actually, when complement is
used on infinite semiring, it is not even guaranteed that the automaton is
finite (\cref{sec:compl}).

\begin{proof}[Sketch of proof of \cref{thm:expaton}, see \cref{app:proof:expaton}]
  This result is proved as \citep[Theorem~4]{lombardy.2005.tcs}: it requires
  several lemmas whose proofs are simple, but long.

  First define the derivation with respect to a word as the repetition of
  derivation with respect to a letter, and prove that
  $\sem{\derivative{u}{\Ed}} = u^{-1}\sem{\Ed}$.

  Second, prove that the set of derivatives of an expression $\Ed$ with
  respect to words is generated by $D(\Ed)$, a set of expressions, called
  \dfn{derived terms}.  The states of the derived-term automaton are not any
  expressions, they are derived terms (and $\Ed$ itself), so the finiteness
  of $D(\Ed)$ implies that of the automaton.

  $D(\Ed)$ admits a simple inductive computation \citep[Definition
  3]{lombardy.2005.tcs}, to which we add:
  \begin{align}
    \notag
    D(\Ed\AND \Fd)
    &\coloneqq \{\Ed_i \AND \Fd_j \mid \forall \Ed_i \in D(\Ed), \forall \Fd_j\in D(\Fd)\}
    \\
    \label{eq:d:compl}
    D(\Ed^c)
    &\coloneqq \{(\lmul{k_1}{\Ed_1} + \cdots + \lmul{k_n}{\Ed_n})^c \mid \forall k_1, \ldots, k_n \in \K, \forall \Ed_1, \ldots, \Ed_n \in D(\Ed)\}
  \end{align}
  If $\Ed$ features no complement, $D(\Ed)$ is trivially finite.  Equation
  \Cref{eq:d:compl} is related to \emph{determinized} expansions
  (\cref{sec:determ}): in essence it dubs (complements of) all potential
  derivatives of $\Ed$ into derived-terms (comparable to going from
  Antimirov's partial derivatives to Brzozowski's derivatives).  On infinite
  semirings, $D(\Ed^c)$ is infinite (more about this in \cref{sec:compl}).
  However, on finite semirings, such as $\mathbb{B}$, it is finite, albeit
  potentially large.

  Finally, prove that $\sem{\Ac_\Ed}(u) = \sem{\Ed}(u)$ for all words
  $u\in A^*$.
\end{proof}

\begin{longenv}
  \begin{Example}[\cref{ex:e2} continued]
    To compute the expansion of $\Ed_2$, one has:
    \newcommand{\dterm}[1]{\fcolorbox{black}{lip}{$#1$}}
    \begin{align*}
      d(\Fd_2)
      & = \bra*{\tfrac{1}{2}}
        \oplus a \odot \left[\Lmul{\tfrac{1}{6}}{a^*}\right]
        \oplus b \odot \left[\Lmul{\tfrac{1}{3}}{b^*}\right]
      \\
      d(\dterm{\Ed_2}) = d(\Fd_2^*)
      & = \bra{\ec(\Fd_2)^*} \oplus \lmul{\ec(\Fd_2)^*}{\ep(\Fd_2) \cdot \Fd_2^*}
      \\
      & = \bra{2}
        \oplus a \odot \left[\Lmul{\tfrac{1}{3}}{\dterm{a^* \, \Ed_2}}\right]
        \oplus b \odot \left[\Lmul{\tfrac{2}{3}}{\dterm{b^* \, \Ed_2}}\right]
    \end{align*}

    \label{ex:e2:aut}
    The derived terms of $\Ed_2$ are $\Ed_2, a^*\Ed_2,$ and $b^*\Ed_2$:
    \begin{align*}
      d(\dterm{a^*\Ed_2})
      & = \bra{2}
        \oplus a \odot \left[\Lmul{\tfrac{4}{3}}{\dterm{a^* \, \Ed_2}}\right]
        \oplus b \odot \left[\Lmul{\tfrac{2}{3}}{\dterm{b^* \, \Ed_2}}\right]
      \\
      d(\dterm{b^*\Ed_2})
      & = \bra{2}
        \oplus a \odot \left[\Lmul{\tfrac{1}{3}}{\dterm{a^* \, \Ed_2}}\right]
        \oplus b \odot \left[\Lmul{\tfrac{5}{3}}{\dterm{b^* \, \Ed_2}}\right]
    \end{align*}

    The derived-term automaton of $\Ed_2$ is therefore:

    \centerline{\includegraphics[scale=.8]{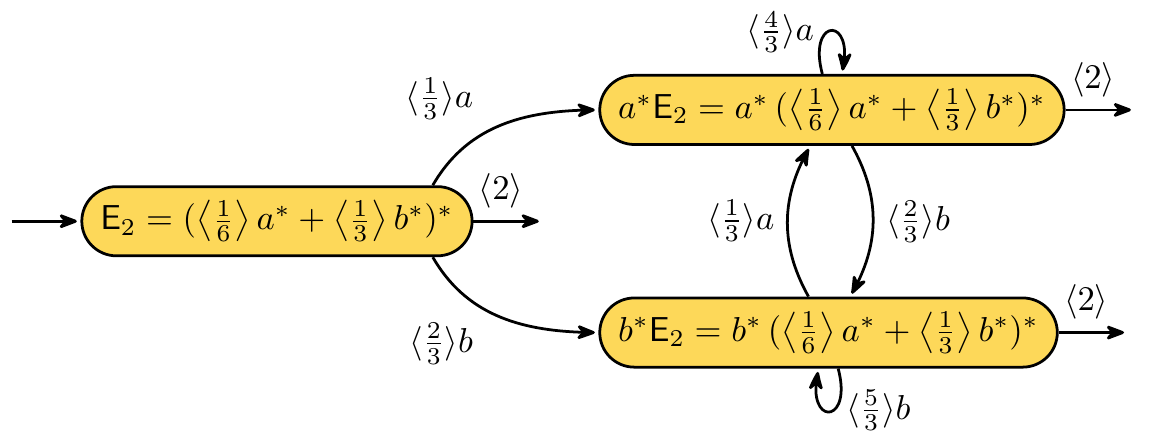}}
  \end{Example}
  \vspace{-10mm}
\end{longenv}

\subsection{Deterministic Automata}
\label{sec:determ}
The exposed approach can be used to generate \emph{deterministic} automata
by \emph{determinizing} the expansions:
$\mathsf{det}(\Xd) \coloneqq \bigoplus_{a\in f(\Xd)}
\Lmul{\unK}{\expr{\Xd_a}}$.
The $\mathsf{expr}$ operator ``consolidates'' a polynomial into an
expression that ensures this determinism.  For instance the expansion
$a\odot[\Lmul{\unK}{b}\oplus\Lmul{\unK}{c}]$, which would yield two
transitions labeled by $a$, one to $b$ and the other to $c$, is determinized
into $a\odot[\Lmul{\unK}{(b+c)}]$, yielding a single transition, to $b+c$.

It is well known that some nondeterministic \emph{weighted} automata have no
deterministic equivalent, in which case determinization loops.  Our
construct is subject to the same condition.  The expression
$\Ed \coloneqq a^*+(\lmul{2}{a})^*$ on the alphabet $\{a\}$ admits an
infinite number of derivatives:
$\derivative{a^n}(\Ed) = a^* \oplus \lmul{2^n}{(\lmul{2}{a})^*}$.  Therefore
our construction of \emph{deterministic} automata would not terminate: the
automaton is locally finite but infinite (and there is no finite
deterministic automaton equivalent to $\Ed$).  However, a lazy
implementation as available in \vcsn{}\cref{foot:url} would uncover the
automaton on demand, for instance when evaluating a word.

\smallskip
\centerline{\includegraphics[scale=.8]{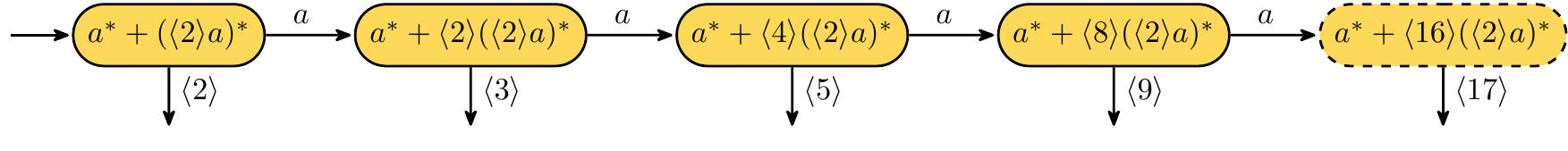}}

To improve determinizability, when $\K$ features a left-division, we apply
the usual technique used in weighted determinization implementations:
normalize the results to keep a unique representative of colinear
polynomials.  Concretely, when determinizing expansions, polynomials are
first normalized:
\begin{math}
  \mathsf{det}(\Xd) \coloneqq \bigoplus_{a\in f(\Xd)} \Lmul{\abs{\Xd_a}}{\expr{\abs{\Xd_a} \backslash \Xd_a}}
\end{math}
where, for a polynomial $\Pd = \bigoplus_{i\in I} \Lmul{k_i}{\Ed_i}$, and a
weight $k$,
$k \backslash \Pd \coloneqq \bigoplus_{i\in I} \Lmul{k \backslash
  k_i}{\Ed_i}$,
and the weight $\abs{\Pd}$ denotes some ``norm'' of (the coefficients of)
$\Pd$.  For instance $\abs{\Pd}$ can be the GCD of the $k_i$ (so that the
coefficients are coprime), or, in the case of a field, the first non null
$k_i$ (so that the first non null coefficient is $\unK$), or the sum of the
$k_i$ provided it's not null (so that the sum of the coefficients is
$\unK$), etc.

\begin{Example}[\cref{ex:e1,ex:e1:xpn,ex:e1:end} cont.]
  \label{ex:e1:det}
  The deterministic derived-term automaton of $\Ed_1$ using
  GCD-normalization is:
  \\[-1ex]
  \centerline{\includegraphics[scale=.8]{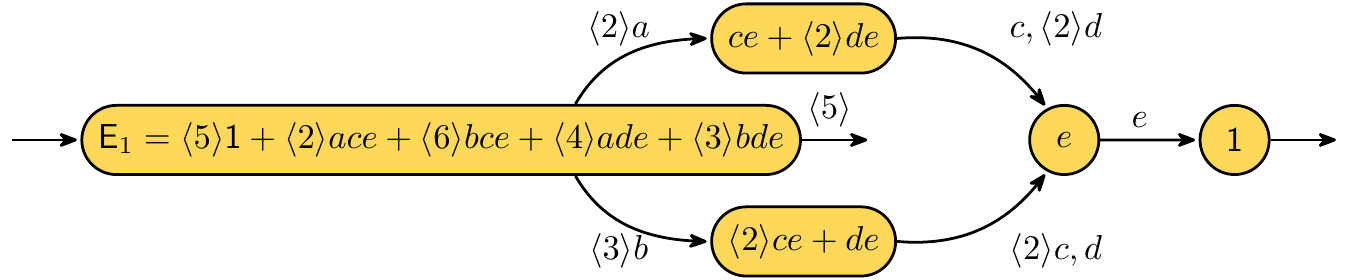}}
\end{Example}
\vspace{-10mm}

\subsection{The Case of Complement}
\label{sec:compl}

It is well known that to complement an (unweighted) automaton, it needs to
be deterministic and complete (which can lead to an exponential number of
states).  ``Local'' determinism (i.e., restricted to complemented
subexpressions) is ensured by $\mathsf{expr}$ in the definition of the
complement of an expansion in \cref{eq:poly:ops:andcompl,eq:epn:compl:epn}.

In the case of weighted expressions, we hit the same problems ---and apply
the same techniques--- as in \cref{sec:determ}: not all expressions generate
finite automata.  A strict (non-lazy) implementation would not terminate on
$\paren{a^*+(\lmul{2}{a})^*}^c$; a lazy implementation would uncover finite
portions of the automaton, on demand.  However, although
$\Fd \coloneqq (\lmul{2}{a})^*+(\lmul{4}{aa})^*$ admits an infinite number
of derivatives, $\Fd^c$ features only two:
$\paren{\lmul{2}{(\lmul{2}{a})^*} + \lmul{4}{(a(\lmul{4}{aa})^*)}}^c
\Rightarrow \paren{(\lmul{2}{a})^* + \lmul{2}{(a(\lmul{4}{aa})^*)}}^c$
and itself.  It is the trivial identity
$(\lmul{k}{\Ed})^c \Rightarrow \Ed^c$ that eliminates the common factor.

\begin{Example}[\cref{ex:ab} continued]
  \label{ex:ab:xpn}%
  We have (see \cref{ex:ab:xpn:detailed} in \cref{sec:appendix} for
  details):
  \begin{align*}
    d(\Ed_3)
    &= a \odot [\Lmul{2}{b} \oplus \Lmul{3}{\paren{b^{c} \AND (a+b)^*}}]
      \oplus b \odot [\Lmul{3}{(a+b)^*}]
  \end{align*}%
  \label{ex:ab:aut}%
  The lower part of $\mathcal{A}_{\Ed_3}$ is characteristic of the
  complement of a complete deterministic automaton:
  \\[-1ex]
  \centerline{\includegraphics[scale=.8]{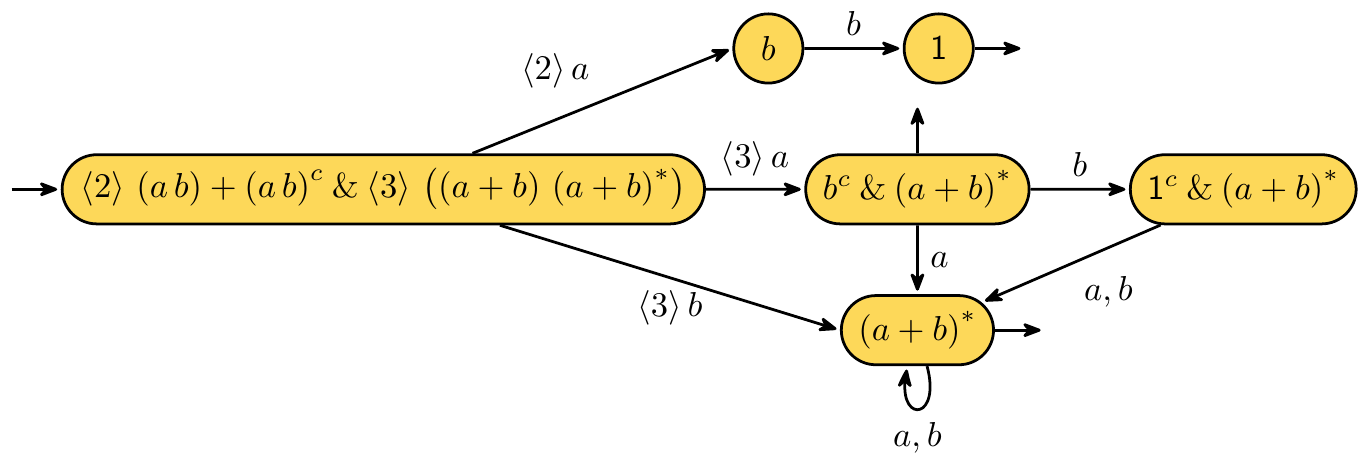}}
\end{Example}
\vspace{-10mm}

\subsection{Complexity and Performances}
\label{sec:perfs}
We focus on basic expressions.  Obviously, $\width{\Ed} \le \length{\Ed}$,
and we know $\length{\Ac_\Ed} \le \width{\Ed} + 1$.

The complexity of Antimirov's algorithm is $O(\width{E}^3\length{E}^2)$
\citep{champarnaud.2002.tcs}: for each of the $\length{\Ac_\Ed}$ states, we
may generate at most $\length{\Ac_\Ed}$ partial derivatives, each one to
compare to the $\length{\Ac_\Ed}$ derived-terms.  That's
$O(\length{\Ac_\Ed}^3)$ comparisons to perform on objects of size
$O(\length{\Ed}^2)$.

However, hash tables allow to avoid these costly comparisons.  For each of
the $\length{\Ac_\Ed}$ states, we may generate at most $\length{\Ac_\Ed}$
partial derivatives and number them via a hash table.  Computing an
expansion builds an object of size $O(\length{\Ed}^2)$, however using
references instead of deep copies allows to stay linear, so the complexity
is $O(\width{E}^2\length{E})$.

\smallskip
To build the derived-term automaton using derivation, one loops over the
alphabet for each derived term.  This incurs a performance penalty with
large alphabets.
%
%
The following table reports the duration of the process, in milliseconds,
for $\Ed_n \coloneqq (a+b)^*a(a+b)^n$ (right associative) by
\vcsn\footnote{\vcsn 2.2 as of 2016-01-29, compiled with Clang 3.6 with
  options \texttt{-O3 -DNDEBUG}, and run on a Mac OS X 10.11.3, Intel Core
  i7 2.9GHz, 8GB of RAM.  Best run out of five.}, depending on $n$, for two
alphabet sizes: 2 and 254 (\vcsn reserves two \texttt{chars}).

\smallskip
\noindent
\newcommand{\T}[1]{\multicolumn{1}{c}{~~~~~#1~~~}}
\begin{tabular}{l@{~~~}d{2}d{2}d{1}d{1}d{0}d{0}d{0}}
  {}             & \T{5} & \T{10} & \T{50} & \T{100} & \T{500} & \T{1000} & \T{5000} \\
  \cmidrule{2-8}
  derivation 2   & 0.08  & 0.12   & 0.80   & 2.5     & 55      & 210      & 4,735    \\
  derivation 254 & 1.12  & 2.15   & 15.56  & 39.2    & 694     & 2,448    & 59,019   \\
  expansion 2    & 0.08  & 0.10   & 0.55   & 1.2     & 20      & 70       & 1,617    \\
  expansion 254  & 0.08  & 0.11   & 0.49   & 1.2     & 19      & 70       & 1,619    \\
\end{tabular}
\vspace*{.2\baselineskip}

Even on a two-letter alphabet, the expansion-based algorithm performs better
than the derivation-based one. (To put things in perspective, the
construction of the standard automaton for $n = 5000$ takes $8.2$ms.)

One can optimize the derivation-based algorithm by computing the firsts
globally \citep{owens.2009.jfp} or locally, on-the-fly, and then derivating
on this set.  However, on sums such as $a_1+\cdots+a_n$ (where $a_i$ are
distinct letters) the expansion requires a single traversal ($O(n)$) whereas
one still needs $n$ derivations, a $O(n^2)$ process.
Besides, the derivation-based algorithm computes the constant term of an
expression several times: to check whether the current state is final, to
compute the derivation of products and stars, and to compute the firsts of
products.  To fix this issue, these repeated computations can be cached.

Addressing both concerns (iteration over the alphabet, repeated computation
of the constant term) for the derivation-based algorithm requires three
tightly entangled algorithms (constant term, derivation, first).
Expansions, on the other hand, keep them together, in a single construct,
computed in a single traversal of the expression.


\section{Related Work}
\label{sec:related}
Compared to \citet{brzozowski.64.jacm} we introduced \emph{weighted}
expansions, and their direct computation, making them the core computation
of the algorithm.  This was partly done for basic Boolean expressions by
\citet{antimirov.1996.tcs} as ``linear forms''.

Aside from our support for weighted expressions, our approach of extended
operators is comparable to that of \citet{caron.2011.lata.2}, but, we
believe, using a simpler framework.  Basically, their sets of sets of
expressions correspond to polynomials of conjunctions: their
$\{\{\Ed, \Fd\}, \{\Gd, \Hd\}\}$ is our $\Ed \AND \Fd \oplus \Gd \AND \Hd$.
Using our framework, the automaton of Fig.~3 \citep{caron.2011.lata.2} has
one state less, since $\{\Ed, \Fd\}$ and $\{\Ed \cap \Fd\}$ both are
$\Ed \AND \Fd$.  Actually, the main point of sets of sets of expressions is
captured by our \emph{distributive} definition of the conjunction of
polynomials, \cref{eq:poly:ops:andcompl}, which matches that of their
\raisebox{1pt}{\scalebox{.7}{\circled{$\cap$}}} operator%
\begin{longenv}
  ; indeed what they call the ``natural extension''
  \citep[Sect.~3.1]{caron.2011.lata.2} would correspond to
  $\Pd_1 \AND \Pd_2 \coloneqq \expr{\Pd_1} \AND \expr{\Pd_2}$
\end{longenv}
.  Additional properties, such as associativity of $\AND$, can be enabled
via additional trivial identities.  Like us, their
\raisebox{1pt}{\scalebox{.7}{\circled{$\neg$}}} operator ensures that
complemented expressions generate deterministic automata.

For basic (weighted) expressions, completely different approaches build the
derived-term automaton with a quadratic complexity
\citep{allauzen.2006.mfcs.2, champarnaud.2007.dlt.2}.  However, the
expansion-based algorithm features some unique properties.  It supports a
simple and natural on-the-fly implementation.  It provides insight on the
built automata by labeling states with the language/series they denote
(e.g., \vcsn renders derived-term automata as in \ifthenelse{\boolean{long}}
{\cref{ex:e1:end,ex:e1:det,ex:e2:aut,ex:ab:aut}}
{\cref{ex:e1:end,ex:e1:det,ex:ab:aut}}).  It is a flexible framework in
which new operators can be easily supported (e.g., the shuffle and
infiltration operators in \vcsn).  It supports the direct construction of
deterministic automata.  And it copes easily with alternative derivation
schemes, such as the ``broken derived-terms'' \citep{lombardy.04.latin,
  lombardy.2005.tcs, lombardy.2010.rairo, angrand.2010.jalc}.

\section{Conclusion}
\label{sec:conc}
The construction of the derived-term automaton from a weighted rational
expression is a powerful technique: states have a natural interpretation
(they are identified by their future: the series they compute), extended
rational expressions are easily supported, determinism can be requested, and
it even offers a natural lazy, on-the-fly, implementation to handle infinite
automata.

To build the derived-term automaton, we generalized Brzozowski's expansions
to weighted expressions, and an inductive algorithm to compute the expansion
of a rational expression.  The formulas on which this algorithm is built
reunite as a unique entity three facets that were kept separated in previous
works: constant term, firsts, and derivatives.  This results in a simpler
set of equations, and an implementation whose complexity is independent of
the size of the alphabet and even applies when it is infinite (e.g., when
labels are strings, integers, etc.).  Building the derived-term automaton
using expansions is straightforward.  Derivatives are only a technical tool
to prove the correctness of the derived-terms.  We have also shown that
using proper techniques, the complexity of the algorithm is much better that
previously reported.

The computation of expansions and derivations are implemented in
\vcsn{}\cref{foot:url}, together with their automaton construction
procedures (possibly lazy, possibly deterministic).  Our implementation
actually prototypes support for additional operators on rational expressions
(e.g., shuffle and infiltration).  Our future work is focused on these
operators.

\subparagraph*{Acknowledgments} Interactions with A.~Duret-Lutz,
S.~Lombardy, L.~Saiu and J.~Sakarovitch resulted in this work.  Anonymous
reviewers made very helpful comments.

\bibliographystyle{myabbrvnat}
\bibliography{%
  article,%
  share/bib/acronyms,%
  share/bib/lrde,%
  share/bib/comp.lang.c++,%
  share/bib/comp.compilers.automata%
}

\begin{thebibliography}{17}
\providecommand{\natexlab}[1]{#1}
\providecommand{\url}[1]{\texttt{#1}}
\expandafter\ifx\csname urlstyle\endcsname\relax
  \providecommand{\doi}[1]{doi: #1}\else
  \providecommand{\doi}{doi: \begingroup \urlstyle{rm}\Url}\fi

\bibitem[Allauzen and Mohri(2006)]{allauzen.2006.mfcs.2}
C.~Allauzen and M.~Mohri.
\newblock A unified construction of the {G}lushkov, follow, and {A}ntimirov
  automata.
\newblock In \emph{MFCS}, vol. 4162 of \emph{LNCS}, pp. 110--121. Springer,
  2006.

\bibitem[Angrand et~al.(2010)Angrand, Lombardy, and
  Sakarovitch]{angrand.2010.jalc}
P.-Y. Angrand, S.~Lombardy, and J.~Sakarovitch.
\newblock On the number of broken derived terms of a rational expression.
\newblock \emph{Journal of Automata, Languages and Combinatorics}, 15\penalty0
  (1/2):\penalty0 27--51, 2010.

\bibitem[Antimirov(1996)]{antimirov.1996.tcs}
V.~Antimirov.
\newblock Partial derivatives of regular expressions and finite automaton
  constructions.
\newblock \emph{TCS}, 155\penalty0 (2):\penalty0 291--319, 1996.

\bibitem[Brzozowski(1964)]{brzozowski.64.jacm}
J.~A. Brzozowski.
\newblock Derivatives of regular expressions.
\newblock \emph{J. ACM}, 11\penalty0 (4):\penalty0 481--494, 1964.

\bibitem[Caron et~al.(2011)Caron, Champarnaud, and Mignot]{caron.2011.lata.2}
P.~Caron, J.-M. Champarnaud, and L.~Mignot.
\newblock Partial derivatives of an extended regular expression.
\newblock In \emph{LATA}, vol. 6638 of \emph{LNCS}, pp. 179--191. Springer,
  2011.

\bibitem[Champarnaud and Ziadi(2002)]{champarnaud.2002.tcs}
J.-M. Champarnaud and D.~Ziadi.
\newblock Canonical derivatives, partial derivatives and finite automaton
  constructions.
\newblock \emph{TCS}, 289\penalty0 (1):\penalty0 137--163, 2002.

\bibitem[Champarnaud et~al.(2007)Champarnaud, Ouardi, and
  Ziadi]{champarnaud.2007.dlt.2}
J.-M. Champarnaud, F.~Ouardi, and D.~Ziadi.
\newblock An efficient computation of the equation $\mathbb{K}$-automaton of a
  regular $\mathbb{K}$-expression.
\newblock In \emph{DLT}, vol. 4588 of \emph{LNCS}. Springer, 2007.

\bibitem[Demaille et~al.(2013)Demaille, Duret-Lutz, Lombardy, and
  Sakarovitch]{demaille.13.ciaa}
A.~Demaille, A.~Duret-Lutz, S.~Lombardy, and J.~Sakarovitch.
\newblock Implementation concepts in {V}aucanson 2.
\newblock In \emph{CIAA'13}, vol. 7982 of \emph{LNCS}, pp. 122--133, July 2013.
  Springer.

\bibitem[Glushkov(1961)]{glushkov.61.rms}
V.~M. Glushkov.
\newblock The abstract theory of automata.
\newblock \emph{Russian {M}ath. {S}urveys}, 16:\penalty0 1--53, 1961.

\bibitem[Lombardy and Sakarovitch(2004)]{lombardy.04.latin}
S.~Lombardy and J.~Sakarovitch.
\newblock How expressions can code for automata.
\newblock In \emph{LATIN}, pp. 242--251, 2004.

\bibitem[Lombardy and Sakarovitch(2005)]{lombardy.2005.tcs}
S.~Lombardy and J.~Sakarovitch.
\newblock Derivatives of rational expressions with multiplicity.
\newblock \emph{TCS}, 332\penalty0 (1-3):\penalty0 141--177, 2005.

\bibitem[Lombardy and Sakarovitch(2010)]{lombardy.2010.rairo}
S.~Lombardy and J.~Sakarovitch.
\newblock Corrigendum to our paper: How expressions can code for automata.
\newblock \emph{{RAIRO} --- Theoretical Informatics and Applications},
  44\penalty0 (3):\penalty0 339--361, 2010.

\bibitem[McNaughton and Yamada(1960)]{mcnaughton.60.itec}
R.~McNaughton and H.~Yamada.
\newblock Regular expressions and state graphs for automata.
\newblock \emph{IEEE Transactions on Electronic Computers}, 9:\penalty0 39--47,
  1960.

\bibitem[Owens et~al.(2009)Owens, Reppy, and Turon]{owens.2009.jfp}
S.~Owens, J.~Reppy, and A.~Turon.
\newblock Regular-expression derivatives re-examined.
\newblock \emph{J. Funct. Program.}, 19\penalty0 (2):\penalty0 173--190, Mar.
  2009.

\bibitem[Rutten(1999)]{rutten.1999.icalp}
J.~J. M.~M. Rutten.
\newblock Automata, power series, and coinduction: Taking input derivatives
  seriously.
\newblock In \emph{26th International Colloquium on Automata, Languages and
  Programming, ICALP'99, Prague, Czech Republic, July 11-15, 1999,
  Proceedings}, vol. 1644 of \emph{LNCS}, pp. 645--654. Springer, 1999.

\bibitem[Rutten(2003)]{rutten.2003.tcs}
J.~J. M.~M. Rutten.
\newblock Behavioural differential equations: a coinductive calculus of
  streams, automata, and power series.
\newblock \emph{TCS}, 308\penalty0 (1-3):\penalty0 1--53, 2003.

\bibitem[Sakarovitch(2009)]{sakarovitch.09.eat}
J.~Sakarovitch.
\newblock \emph{Elements of Automata Theory}.
\newblock Cambridge University Press, 2009.
\newblock Corrected English translation of \emph{\'El\'ements de th\'eorie des
  automates}, Vuibert, 2003.

\end{thebibliography}

\appendix
\section{Appendix}
\label{sec:appendix}

\begin{proof}[Proof of \cref{lem:xpn:semantics}]
  Most operators are trivial, we focus here on the extended operators.
  \begin{align*}
    \sem{\Xd \AND \Yd}
    & = \sem{\bra{\Xd_\eword\Yd_\eword}
      \oplus
      \bigoplus_{\mathclap{a \in f(\Xd) \cap f(\Yd)}} a \odot [\Xd_a \AND \Yd_a]}
    & \text{by definition, \cref{eq:epn:and:epn}}
    \\
    & = \Xd_\eword\Yd_\eword
      + \sum_{\mathclap{a \in f(\Xd) \cap f(\Yd)}} a \cdot \sem{\Xd_a \AND \Yd_a}
    & \text{by definition of $\mathsf{expr}$}
    \\
    & = \Xd_\eword\Yd_\eword
      + \sum_{\mathclap{a \in f(\Xd) \cap f(\Yd)}} a \cdot \left(\sem{\Xd_a} \AND \sem{\Yd_a}\right)
    & \text{by \cref{lem:poly:ops}}
    \\
    & = \paren{\Xd_\eword + \sum_{\mathclap{a \in f(\Xd)}} a \cdot \sem{\Xd_a}}
      \AND
      \left(\Yd_\eword + \sum_{\mathclap{a \in f(\Yd)}} a \cdot \sem{\Yd_a}\right)
    & \text{by \cref{eq:series:and:zip}}
    \\
    & = \sem{\bra{\Xd_\eword} \oplus \bigoplus_{\mathclap{a \in f(\Xd)}} a \cdot [\Xd_a]}
      \AND
      \sem{\bra{\Yd_\eword} \oplus \bigoplus_{\mathclap{a \in f(\Yd)}} a \cdot [\Yd_a]}
    \\
    & = \sem{\Xd} \AND \sem{\Yd}
  \end{align*}

  \begin{align*}
    \sem{\Xd^{c}}
    & = \sem{\bra{\Xd_\eword^{c}}
      \oplus \bigoplus_{a \in f(\Xd)} a \odot [\Xd_a^{c}] \oplus \bigoplus_{a \in A\setminus f(\Xd)} a \odot [\zed^{c}]}
    & \text{by definition, \cref{eq:epn:compl:epn}}
    \\
    & = \Xd_\eword^{c}
      + \sum_{a \in f(\Xd)} a \cdot \sem{\Xd_a^{c}} + \sum_{a \in A\setminus f(\Xd)} a \cdot \sem{\zed^{c}}
    \\
    & = \Xd_\eword^{c}
      + \sum_{a \in f(\Xd)} a \cdot \sem{\Xd_a}^{c} + \sum_{a \in A\setminus f(\Xd)} a \cdot \sem{\zed}^{c}
    & \text{by \cref{lem:poly:ops}}
    \\
    & = \paren{\Xd_\eword + \sum_{a \in f(\Xd)} a \cdot \sem{\Xd_a}}^{c}
    & \text{by \cref{eq:series:compl}}
    \\
    & = \sem{\Xd}^{c}
    && \hfill\qedhere
  \end{align*}
\end{proof}

\begin{Example}[\cref{ex:ab:xpn} detailed]
  \label{ex:ab:xpn:detailed}
  We have:
  \begin{align*}
    \xpa{(ab)^{c}}
     = \xpa{ab}^{c}
     = \paren{\Lmul{a}{[b]}}^{c}
     = \bra{\unK} \oplus a \odot [b^{c}] \oplus b \odot [\zed^{c}]
  \end{align*}
  \begin{align*}
    \xpa{\lmul{3}{(a+b)(a+b)^*}}
    & = \lmul{3}{\xpa{(a+b)(a+b)^*}}\\
    & = \lmul{3}{\xpap{a+b}\cdot(a+b)^* \oplus \bra{\xpae{a+b}}\xpa{(a+b)^*}} \\
    & = \lmul{3}{\paren{a \odot [\Lmul{1}{\und}] \oplus b \odot [\Lmul{1}{\und}]}\cdot(a+b)^* \oplus \bra{\zeK}\xpa{(a+b)^*}} \\
    & = \lmul{3}{\paren{a \odot [\Lmul{1}{(a+b)^*}] \oplus b \odot [\Lmul{1}{(a+b)^*}]}} \\
    & = a \odot [\Lmul{3}{(a+b)^*}] \oplus b \odot [\Lmul{3}{(a+b)^*}]
  \end{align*}
  therefore:
  \begin{align*}
    \Xd
    & \coloneqq \xpa{(ab)^{c}} \AND \xpa{\lmul{3}{(a+b)(a+b)^*}}\\
    & = \paren{\bra{\unK} \oplus a \odot {[b^{c}] \oplus b \odot [\zed^{c}]}}
      \AND \paren{\Lmul{a}{[\Lmul{3}{(a+b)^*}]} \oplus \Lmul{b}{[\Lmul{3}{(a+b)^*}]}}\\
    & = a \odot [\Lmul{3}{\paren{b^{c} \AND (a+b)^*}}]
      \oplus b \odot [\Lmul{3}{\paren{\zed^{c} \AND (a+b)^*}}]
    \\
    & = a \odot [\Lmul{3}{\paren{b^{c} \AND (a+b)^*}}]
      \oplus b \odot [\Lmul{3}{(a+b)^*}]
  \end{align*}
  and finally
  \begin{align*}
    d(\Ed_3)
    & = \xpa{\lmul{2}{ab} + (ab)^{c} \AND \lmul{3}{(a+b)(a+b)^*}} \\
    & = \xpa{\lmul{2}{ab}} \oplus \xpa{(ab)^{c} \AND \lmul{3}{(a+b)(a+b)^*}}\\
    & = a \odot [\Lmul{2}{b}] \oplus \overbrace{\paren{\xpa{(ab)^{c}} \AND \xpa{\lmul{3}{(a+b)(a+b)^*}}}}^{\Xd} \\
    & = a \odot [\Lmul{2}{b}]
      \oplus a \odot [\Lmul{3}{\paren{b^{c} \AND (a+b)^*}}]
      \oplus b \odot [\Lmul{3}{(a+b)^*}]\\
    & = a \odot [\Lmul{2}{b} \oplus \Lmul{3}{\paren{b^{c} \AND (a+b)^*}}]
      \oplus b \odot [\Lmul{3}{(a+b)^*}]
  \end{align*}
\end{Example}

\section{Appendix: Proof of \cref{thm:expaton}}
\label{app:proof:expaton}

Proving this theorem requires several auxiliary results.  None of them is
needed in an implementation: \cref{def:expa-of-expr} is all that is needed
to build the derived-term automaton.

The path, paved by \citet{lombardy.2005.tcs}, is as follows.  First, define
derivation with respect to a word, and show that it is a syntactic
``implementation'' of left-quotient of a series by a word
(\cref{sec:der:word}).  Then define (syntactically) the set of derived
terms, and show that they generate all the word derivatives
(\cref{sec:der-terms}).  Finally show that computations in the derived-term
automaton correspond to computing the left-quotient of the denoted series
(\cref{sec:dt-aut}).

This is also the path followed by the rather terse proof of
\citet[Proposition~4]{caron.2011.lata.2}, but filling the gaps.

\subsection{Derivation by Words}
\label{sec:der:word}
\begin{Definition}[Derivation of a Polynomial]
  $\da{\oplus_{i \in I} \Lmul{k_i}{\Ed_i}}
  \coloneqq
  \oplus_{i \in I} \Lmul{k_i}{\da{\Ed_i}}$
\end{Definition}

\begin{Lemma}\label{lem:der:poly}
  \abovedisplayskip=\abovedisplayshortskip
  \begin{align}
    \label{eq:der:poly:and}
    \da{\paren{\Pd \AND \Qd}} &= \da{\Pd} \AND \da{\Qd} \\
    \label{eq:der:poly:expr}
    \da{\expr{\Pd}} &= \expr{\da{\Pd}}\\
    \label{eq:der:poly:compl}
    \da{\paren{\Pd^c}} &= \paren{\da{\Pd}}^c
  \end{align}
\end{Lemma}

\begin{proof}
  Let
  $\Pd \coloneqq \bigoplus_{i\in I}\bra{k_i} \odot \Ed_i, \Qd \coloneqq
  \bigoplus_{j\in J}\bra{h_j} \odot \Fd_j$.
  \begin{align*}
    \da{\paren{\Pd \AND \Qd}}
    &= \da{\paren{\bigoplus_{i\in I, j\in J} \Lmul{k_ih_j}{\Ed_i \AND \Fd_j}}}  & \text{by def. of polynomial conjunction}\\
    &= \bigoplus_{i\in I, j\in J} \Lmul{k_ih_j}{\da{\paren{\Ed_i \AND \Fd_j}}}\\
    &= \bigoplus_{i\in I, j\in J} \Lmul{k_ih_j}{\paren{\da{\Ed_i} \AND \da{\Fd_j}}} & \text{by \cref{eq:der:and}}\\
    &= \da{\Pd} \AND \da{\Qd} & \text{by def. of polynomial conjunction}
  \end{align*}

  \begin{align*}
    \da{\expr{\Pd}}
    &= \da{\expr{\bigoplus_{i\in I} \Lmul{k_i}{\Ed_i}}}\\
    &= \da{\sum_{i\in I} \lmul{k_i}{\Ed_i}}\\
    &= \sum_{i\in I} \lmul{k_i}{\da{\Ed_i}}\\
    &= \expr{\bigoplus_{i\in I} \Lmul{k_i}{\da{\Ed_i}}}\\
    &= \expr{\da{\Pd}}
  \end{align*}

  \begin{align*}
    \da{\paren{\Pd^c}}
    &= \da{\paren{\expr{\Pd}^c}}         & \text{by def. of polynomial complement} \\
    &= \paren{\da{\paren{\expr{\Pd}}}}^c & \text{by \cref{eq:der:compl}} \\
    &= \paren{\expr{\da{\Pd}}}^c         & \text{by \cref{eq:der:poly:expr}} \\
    &= \paren{\da{\Pd}}^c                & \text{by def. of polynomial complement}
    & \hfill\qedhere
  \end{align*}
\end{proof}

Derivation wrt a single-letter word is defined as the derivation wrt that
letter.  Derivation wrt to a longer word is the result of repeated
derivations wrt letters.
\begin{Definition}[Derivation wrt a Word]
  $\forall a \in A, u \in A^+, \derivative{ua}{\Ed} \coloneqq
  \derivative{a}{\derivative{u}{\Ed}}$.
\end{Definition}

\begin{Lemma}
  \label{lem:der:words}
  $\derivative{uv}{\Ed} = \derivative{v}{\derivative{u}{\Ed}}$
\end{Lemma}

Explicit formulas exist for derivation with respect to a word.
\begin{Lemma}[Direct Computations of Derivation wrt a Word]
  \abovedisplayskip=\abovedisplayshortskip
  \label{lem:der:word}
  \begin{align}
    \du{(\Ed + \Fd)} &= \du{\Ed} \oplus \du{\Fd},
    \\
    \du{(\lmul{k}{\Ed})} &= \lmul{k}{\paren{\du{\Ed}}},
    \\
    \du{(\rmul{\Ed}{k})} &= \rmul{\paren{\du{\Ed}}}{k},
    \\
    \du{(\Ed \cdot \Fd)} &= \paren{\du{\Ed}} \cdot \Fd \oplus \paren{\bigoplus_{\substack{f=gh\\g\in A^*, h\in A^+}}c(\dg{\Ed})\dH{\Fd}}
    \\
    \du{\Ed^*}
    &= \bigoplus_{\substack{f=g_1g_2\cdots g_n\\g_1,\ldots, g_n\in A^+}}
      \lmul{\paren{\prod_{i\in [n-1]}c(\Ed)^* c(\derivative{g_i}{\Ed})}c(\Ed)^*}{\derivative{g_n}{\Ed}\cdot\Ed^*}
    \\
    \label{eq:der:word:and}
    \du{(\Ed \AND \Fd)} &= \du{\Ed} \AND \du{\Fd},
    \\
    \label{eq:der:word:compl}
    \du{\Ed^c} &= \paren{\du{\Ed}}^{c}
  \end{align}
\end{Lemma}

\begin{proof}
  The proof is the same as that of \citep[Prop.~3]{lombardy.2005.tcs}, with
  additional cases for conjunction and complement.

  For conjunction:
  \begin{align*}
    \derivative{ua}{\paren{\Ed \AND \Fd}}
    &= \da{\du{\paren{\Ed \AND \Fd}}} \\
    &= \da{\paren{\du{\Ed} \AND \du{\Fd}}} & \text{by induction hypothesis}\\
    &= \da{\du{\Ed}} \AND \da{\du{\Fd}} & \text{by \cref{eq:der:poly:and}} \\
    &= \derivative{ua}{\Ed} \AND \derivative{ua}{\Fd}
  \end{align*}

  For complement:
  \begin{align*}
    \derivative{ua}{\paren{\Ed^c}}
    &= \da{\du{\paren{\Ed^c}}} \\
    &= \da{\paren{\du{\Ed}}^c} & \text{by induction hypothesis} \\
    &= \paren{\da{\paren{\du{\Ed}}}}^c & \text{by \cref{eq:der:poly:compl}} \\
    &= \paren{\derivative{ua}{\Ed}}^ c
    && \hfill\qedhere
  \end{align*}
\end{proof}

The following lemma makes explicit the connection between the (syntactic)
derivation, and the semantics of an expression.

\begin{Lemma}[{\citep[Prop.~4]{lombardy.2005.tcs}}]
  \label{lem:sem:der:word}
  $\forall u \in A^+, \sem{\Ed}(u) = c(\derivative{u}{\Ed})$.
\end{Lemma}

\begin{proof}
  For conjunction:
  \begin{align*}
    \sem{\Ed\AND\Fd}(u)
    &= \paren{\sem{\Ed}\AND\sem{\Fd}}(u) & \text{by definition} \\
    &= \sem{\Ed}(u) \cdot \sem{\Fd}(u)   & \text{by definition} \\
    &= c(\du{\Ed}) \cdot c(\du{\Fd})     & \text{by induction hypothesis} \\
    &= c(\du{\Ed} \AND \du{\Fd})         & \text{by \cref{eq:der:and}} \\
    &= c(\du{\paren{\Ed \AND \Fd}})      & \text{by \cref{eq:der:word:and}}
  \end{align*}

  For complement:
  \begin{align*}
    \sem{\Ed^c}(u)
    &= \sem{\Ed}^c(u)    & \text{by definition} \\
    &= (\sem{\Ed}(u))^c  & \text{by definition} \\
    &= (c(\du{\Ed}))^c   & \text{by induction hypothesis} \\
    &= c((\du{\Ed})^c)   & \text{by \cref{eq:der:compl}} \\
    &= c(\du{(\Ed^c)})   & \text{by \cref{eq:der:word:compl}}
    & \hfill\qedhere
  \end{align*}
\end{proof}

The previous lemma allows to show the connection between the (syntactic)
derivation, and the (semantical) left-quotient of a series.
\begin{theorem}[{\citep[Theorem~1]{lombardy.2005.tcs}}]
  \label{thm:sem:derivative}
  $\forall u \in A^+, \sem{\du{\Ed}} = u^{-1}\sem{\Ed}$.
\end{theorem}
\begin{proof}
  For any word $v \in A^+$,
  \begin{align*}
    \sem{\du{\Ed}}(v)
    &= c(\dv{\du{\Ed}})        & \text{by \cref{lem:sem:der:word}} \\
    &= c(\derivative{uv}{\Ed}) & \text{by \cref{lem:der:words}} \\
    &= \sem{\Ed}(uv)           & \text{by \cref{lem:sem:der:word}} \\
    &= (u^{-1}\sem{\Ed})(v)    & \text{by definition of left-quotient}
    & \hfill\qedhere
 \end{align*}
\end{proof}

\subsection{Derived Terms}
\label{sec:der-terms}
\begin{Definition}[Derived Terms]
  \label{def:dts}
  Given an expression $\Ed$, its \dfn{derived terms} is the set $D(\Ed)$
  defined as follows:
  \begin{align*}
    D(\zed) &\coloneqq \emptyset \\
    D(\und) &\coloneqq \emptyset \\
    D(a)    &\coloneqq \{\und\} \ee \forall a \in A\\
    D(\Ed + \Fd) &\coloneqq D(\Ed) \cup D(\Fd) \\
    D(\lmul{k}{\Ed}) &\coloneqq D(\Ed) \ee \forall k \in \K \\
    D(\rmul{\Ed}{k}) &\coloneqq \{\rmul{\Ed_i}{k} \mid \Ed_i \in D(\Ed)\}  \ee \forall k \in \K \\
    D(\Ed\cdot \Fd) &\coloneqq  \{\Ed_i\cdot\Fd \mid \Ed_i \in D(\Ed)\} \cup D(\Fd) \\
    D(\Ed^*) &\coloneqq  \{\Ed_i\cdot\Ed^* \mid \Ed_i \in D(\Ed)\} \\
    D(\Ed\AND \Fd)
    &\coloneqq \{\Ed_i \AND \Fd_j \mid \forall \Ed_i \in D(\Ed), \forall \Fd_j\in D(\Fd)\}
    \\
    D(\Ed^c)
    &\coloneqq \{(\lmul{k_1}{\Ed_1} + \cdots + \lmul{k_n}{\Ed_n})^c \mid \forall k_1, \ldots, k_n \in \K, \forall \Ed_1, \ldots, \Ed_n \in D(\Ed)\}
  \end{align*}
  \noindent where in the last equation, the $\Ed_i$ are sorted.  Besides,
  depending on the features of $\K$, the coefficients may be normalized so
  that colinear combinations are represented only once.  For instance if
  $\K$ has no zero divisor, one may divide by the GCD of the $k_i$ (so that
  the $k_i$ are coprime), or, in the case of a field, by the first non null
  $k_i$ (so that the first non null coefficient is $\unK$), or by the sum of
  the $k_i$ provided it's not null (so that the sum of the coefficients is
  $\unK$), etc.
\end{Definition}

\begin{theorem}
  If $\K$ is finite, or if $\Ed$ has no complement, then $D(\Ed)$ is finite.
\end{theorem}

\begin{proof}
  This is a direct consequence from \cref{def:dts}: finiteness propagates
  during the induction.  The only danger is the case of complement, whose
  finiteness ensues from a very crude criterion: there exists a finite
  number of combinations.
\end{proof}

We prove that the set of derived terms is closed by derivation.  The
insightful reader can see automata dawning: the derived terms are the
states, and the coefficients are the weights of the transitions.
\begin{Lemma}
  \label{lem:dt:closed:letter}
  We denote $\{1, \ldots, n\}$ by $[n]$.

  Let $\Ed$ be an expression, $D(\Ed) = \{\Ed_i \mid i \in [n]\}$ be its
  derived terms.  There exists $n$ coefficients $(k_i^{(a)})_{i\in[n]}$ and
  $n^2$ coefficients $(k_{i,j}^{(a)})_{i,j\in[n]}$ such that
  \begin{gather*}
    \da{\Ed}   = \bigoplus_{i\in[n]} \lmul{k_i^{(a)}}{\Ed_i}
    \qquad
    \da{\Ed_i} = \bigoplus_{i'\in[n]} \lmul{k_{i,i'}^{(a)}}{\Ed_{i'}}
  \end{gather*}
\end{Lemma}

\begin{proof}
  We follow \citep[proof of Theorem~2]{lombardy.2005.tcs}, to which we add
  the following cases.  We note:
  \begin{gather*}
    D(\Fd)     = \{F_j \mid j \in [m]\} \qquad
    \da{\Fd}   = \bigoplus_{j\in[m]}  \lmul{h_j^{(a)}}{\Fd_j} \qquad
    \da{\Fd_j} = \bigoplus_{j'\in[m]} \lmul{h_{j,j'}^{(a)}}{\Fd_{j'}}
  \end{gather*}

  Consider $\Ed \AND \Fd$:
  \begin{align*}
    \da{(\Ed\AND\Fd)}
    &= \da{\Ed} \AND \da{\Fd} \\
    &= \paren{\bigoplus_{i\in[n]} \lmul{k_i^{(a)}}{\Ed_i}}
      \AND \paren{\bigoplus_{j\in[m]} \lmul{h_j^{(a)}}{\Fd_j}} \\
    &= \bigoplus_{i\in[n], j\in[m]}
      \lmul{k_i^{(a)}h_j^{(a)}}{\paren{\Ed_i\AND\Fd_j}}
  \end{align*}
  which is indeed a linear combination of derived terms of $\Ed\AND\Fd$,
  since
  $D(\Ed \AND \Fd) = \{\Ed_i\AND\Fd_j \mid \forall \Ed_i \in D(\Ed), \forall
  \Fd_j\in D(\Fd)\}$ by definition \cref{def:dts}.

  Likewise,
  \begin{align*}
    \da{(\Ed_i\AND\Fd_j)}
    &= \da{\Ed_i} \AND \da{\Fd_j} \\
    &= \paren{\bigoplus_{i'\in[n]} \lmul{k_{i,i'}^{(a)}}{\Ed_{i'}}}
      \AND \paren{\bigoplus_{j'\in[m]} \lmul{h_{j,j'}^{(a)}}{\Fd_{j'}}} \\
    &= \bigoplus_{i'\in[n], j'\in[m]}
      \lmul{k_{i,i'}^{(a)}h_{j,j'}^{(a)}}{\paren{\Ed_{i'}\AND\Fd_{j'}}}
  \end{align*}
  is a linear combination of elements of $D(\Ed\AND\Fd)$.

  Consider $\Ed^c$:
  \begin{align*}
    \da{(\Ed^c)}
    &= (\da{\Ed})^c \\
    &= \paren{\bigoplus_{i\in[n]} \lmul{k_i^{(a)}}{\Ed_i}}^c \\
    &= \paren{\expr{\bigoplus_{i\in[n]} \lmul{k_i^{(a)}}{\Ed_i}}}^c \\
    &= \paren{\sum_{i\in[n]} \lmul{k_i^{(a)}}{\Ed_i}}^c
  \end{align*}
  which is a member of $D(\Ed^c)$.  Note in this case, we expect the $\Ed_i$
  to be sorted in the same order as the one used by $\mathsf{expr}$.

  Besides:
  \begin{align*}
    \da{\paren{\paren{\sum_{i\in[n]} \lmul{k_i^{(a)}}{\Ed_i}}^c}}
    &= \paren{\da{\paren{\sum_{i\in[n]} \lmul{k_i^{(a)}}{\Ed_i}}}}^c \\
    &= \paren{\bigoplus_{i\in[n]} \lmul{k_i^{(a)}}{\da{\Ed_i}}}^c\\
    &= \paren{\bigoplus_{i\in[n]} \lmul{k_i^{(a)}}{\bigoplus_{i'\in[n]} \lmul{k_{i,i'}^{(a)}}{\Ed_{i'}}}}^c\\
    &= \paren{\bigoplus_{i,i'\in[n]} \lmul{k_i^{(a)}k_{i,i'}^{(a)}}{\Ed_{i'}}}^c\\
    &= \paren{\sum_{i,i'\in[n]} \lmul{k_i^{(a)}k_{i,i'}^{(a)}}{\Ed_{i'}}}^c
  \end{align*}
  which is a member of $D(\Ed^c)$.
\end{proof}

The following result, similar to \citep[Theorem~3]{lombardy.2005.tcs}, shows
that any word derivative of an expression is a linear combination of its
derived terms.
\begin{theorem}
  \label{thm:k}
  Let $\Ed$ be an expression, $D(\Ed) = \{\Ed_i \mid i \in [n]\}$ be its
  derived terms, and $u \in A^+$ any word.  There exist coefficients
  $(k_i^{(u)})_{i\in[n]}$ in $\K$ such that:
  \begin{align*}
    \du{\Ed}   &= \bigoplus_{i\in[n]} \lmul{k_i^{(u)}}{\Ed_i}
  \end{align*}
\end{theorem}

\begin{proof}
  The result is proved by induction.

  The base case is established by \cref{lem:dt:closed:letter}.
  \begin{align*}
    \derivative{ua}{\Ed}
    &= \da{\du{\Ed}} \\
    &= \da{\paren{\bigoplus_{i\in[n]} \lmul{k_i^{(u)}}{\Ed_i}}}
    & \text{by induction hypothesis}\\
    &= \bigoplus_{i\in[n]} \lmul{k_i^{(u)}}{\da{\Ed_i}} \\
    &= \bigoplus_{i\in[n]} \lmul{k_i^{(u)}}{\paren{\bigoplus_{j\in[n]} \lmul{k_{i,j}^{(a)}}{\Ed_{j}}}}
    & \text{by \cref{lem:dt:closed:letter}} \\
    &= \bigoplus_{j\in[n]} \paren{\bigoplus_{i\in[n]} \lmul{k_i^{(u)}k_{i,j}^{(a)}}{\Ed_{j}}}\\
    &= \bigoplus_{j\in[n]} \lmul{\sum_{i\in[n]} k_i^{(u)}k_{i,j}^{(a)}}{\Ed_{j}}
  \end{align*}
  \noindent
  i.e.,
  \begin{align}
    \label{eq:word:ki}
    k_{j}^{(ua)} = \sum_{i\in[n]} k_i^{(u)}k_{i,j}^{(a)}
  \end{align}
\end{proof}

\subsection{Derived-term Automaton}
\label{sec:dt-aut}
In order to prove the final result, we express automata in a different way
\citep[Sect.~5]{lombardy.2005.tcs}.

\begin{Definition}[Representations of a Finite Weighted Automaton]
  The \dfn{matrix representation} of a (finite weighted) automaton is the
  sextuplet $\bra{A, \K, E, Q, E, I, T}$ where:
  \begin{itemize}
  \item $A$ is an alphabet
  \item $\K$ (the set of weights) is a semiring,
  \item $Q$ is a finite set of states,
  \item $I$ (resp. $T$) is a row (resp. column) vector of dimension $Q$ with
    entries in $\K$,
  \item $E$ is a square matrix whose entries are linear combinations of
    letters of $A$ with coefficients in $\K$.
  \end{itemize}

  \medskip

  The \emph{$\K$-representation} of an automaton is the triple
  $\bra{I, \zeta, T}$ where $\zeta$ is a morphism from $A$ to
  $\K^{Q \times Q}$ such that $E = \sum_{a\in A}\zeta(a)a$.
\end{Definition}

One can then prove that, for every word $u\in A^*$:
\begin{gather*}
  \sem{\Ac}(u)
  = (I \cdot E^* \cdot T)(u)
  = (I \cdot E^{\length{u}} \cdot T)(u)
  = I \cdot \zeta(u) \cdot T
\end{gather*}

Put together, the definition of derivation and constant terms
(\cref{def:ctder}), their connection with expansions (\cref{prop:expa:der}),
the definition of $\Ac_\Ed$, the expansion-based derived-term automaton of
$\Ed$ (\cref{def:expaton}), and finally \cref{lem:dt:closed:letter}, show
that $\Ac_\Ed$ admits the following $\K$-representation:
\begin{gather*}
  I_{\Ed_i}
  =
    \begin{cases}
      \unK & \text{if $\Ed_i = \Ed$} \\
      \zeK & \text{otherwise}
    \end{cases}
  \qquad
  \zeta(a)_{i,j} = k_{i, j}^{(a)}
  \qquad
  T_{\Ed_i}
   = c(\Ed_i)
\end{gather*}

\noindent
where the coefficients $k_{i, j}^{(a)}$ were defined in
\cref{lem:dt:closed:letter}.  The $\Ed_i$ are the derived-terms of $\Ed$, to
which we add $\Ed_0 \coloneqq \Ed$ if $\Ed \not\in D(\Ed)$, in which case
$k_{i,0}^{(a)} \coloneqq \zeK$, and $k_{0,i}^{(a)} \coloneqq k_i^{(a)}$ for
all $i > 0$.

We prove by induction that:
\begin{gather}
  \label{eq:i-zeta}
  \forall u \in A^+, \forall i \in [n],  (I\cdot \zeta(u))_i = k_i^{(u)}
\end{gather}

\begin{proof}
  The base case:
  \begin{align*}
    (I \cdot \zeta(a))_i
    &= \sum_j (I_j \cdot \zeta(a)_{j,i} \\
    &= \unK \cdot \zeta(a)_{0,i} & \text{by definition of $I$} \\
    &= k_{0,i}^{(a)} & \text{by definition of $\zeta$} \\
    &= k_i^{(a)} & \text{by definition of $k_{i,j}^{(a)}$}
  \end{align*}

  Then the induction:
  \begin{align*}
    (I\cdot \zeta(ua))_i
    &= (I\cdot (\zeta(u)\cdot\zeta(a)))_i \\
    &= ((I\cdot \zeta(u))\cdot\zeta(a))_i \\
    &= \sum_{j}(I\cdot \zeta(u))_j \cdot \zeta(a)_{j,i}) \\
    &= \sum_{j}(k_j^{(u)} \cdot \zeta(a)_j)_i & \text{by induction hypothesis}\\
    &= \sum_{j}(k_j^{(u)} \cdot k^{(a)}_{j,i}) & \text{by definition of $\zeta$}\\
    &= k_i^{(ua)} & \text{by \cref{eq:word:ki}}
    & \hfill\qedhere
  \end{align*}
\end{proof}

We can now finally prove that $\sem{\Ac_\Ed} = \sem{\Ed}$.  Let $u \in A^+$:
\begin{align*}
  \sem{\Ac_\Ed}(u)
  &= (I \cdot \zeta(u) \cdot T) \\
  &= \sum_i (I \cdot \zeta(u))_i \cdot T_i \\
  &= \sum_i k_i^{(u)} \cdot T_i & \text{by \cref{eq:i-zeta}} \\
  &= \sum_i k_i^{(u)} \cdot c(\Ed_i) & \text{by definition of $T$} \\
  &= c\paren{\bigoplus_i \lmul{k_i^{(u)}}{\Ed_i}}\\
  &= c\paren{\du{\Ed}} & \text{by \cref{thm:k}}\\
  &= \sem{\Ed}(u) & \text{by \cref{thm:sem:derivative}} \\
\end{align*}

The case of the empty word follows from the definition of $I$ and $T$:
$\sem{\Ac_\Ed}(\eword) = \sum_i (I_i . T_i) = \unK \cdot T_0 = c(\Ed)$.
\end{document}

